\documentclass[11pt]{article} 

\usepackage[a4paper,left=2.5cm,right=2.5cm,top=2.5cm,bottom=3.0cm]{geometry}
\usepackage[dvips]{epsfig}
\usepackage{amsfonts}
\usepackage{moreverb}

\usepackage{amsthm}

\newtheorem{theorem}{Theorem}[section]
\newtheorem{lemma}{Lemma}[section]
\newtheorem{corollary}{Corollary}[section]

\newtheorem{proposition}{Proposition}[section]

\usepackage{graphicx}
\usepackage{xcolor}
\usepackage{enumitem}
\usepackage{amssymb}
\usepackage{amsmath}

\usepackage[noend, noline, ruled, linesnumbered]{algorithm2e} %allowed
\SetAlFnt{\normalsize}
\DontPrintSemicolon
\SetKwComment{Comment}{$\triangleright$\ }{}
\SetKwProg{Def}{def}{:}{}
\SetArgSty{textnormal}
\SetKwRepeat{Do}{do}{while}
\SetKwComment{Comment}{$\triangleright$\ }{}

\begin{document}
%%%%%%%%%%%%%%%%%%%%%%%%%%%%%%%%%%%%%%%%%%%%%%%%%%%%%%%%%%%

\title{\bf Exploring Wedges of an Oriented Grid\\by an Automaton with Pebbles}

\author{
Subhash Bhagat\footnotemark[1]
\and
Andrzej Pelc\footnotemark[1]\footnotemark[2]
}
\date{ }
\maketitle
\def\thefootnote{\fnsymbol{footnote}}

\footnotetext[1]{Department of Mathematics, Indian Institute of Technology Jodhpur, India. {\tt sbhagat@iitj.ac.in}}

\footnotetext[2]{D\'{e}partement d'informatique, Universit\'{e} du Qu\'{e}bec en Outaouais, Canada. {\tt pelc@uqo.ca}\\ Supported in part by NSERC discovery grant 2018-03899 and by
the Research Chair in Distributed Computing of the
Universit\'{e} du Qu\'{e}bec en Outaouais.}

\centerline{\today}

%%%%%%%%%%%%%%%%%%

\maketitle

\begin{abstract}
A mobile agent, modeled as a deterministic finite automaton, navigates in the infinite anonymous oriented grid $\mathbb{Z} \times \mathbb{Z}$.   
It has to explore a given infinite subgraph of the grid by visiting all of its nodes. We focus on the simplest subgraphs, called {\em wedges}, spanned by all nodes of the grid located between two
half-lines in the plane, with common origin. Many wedges turn out to be impossible to explore by an automaton that cannot mark nodes of the grid. Hence, we study the following question:
Given a wedge $W$, what is the smallest number $p$ of (movable) pebbles for which there exists an automaton that can explore $W$ using $p$ pebbles? Our main contribution is a complete solution of this problem. For each wedge $W$ we determine this minimum number $p$, show an automaton that explores it using $p$ pebbles and show that fewer pebbles are not enough.
We show that this smallest number of pebbles can vary from 0 to 3, depending on the angle between half-lines limiting the wedge and depending on whether the automaton can cross these half-lines or not. 
\vspace{1ex}

\end{abstract}

\thispagestyle{empty}
\newpage
\setcounter{page}{1}

\baselineskip    0.2in
\parskip         0.0in
\parindent       0.3in

%%%%%%%%%%%%%%%%%%%%%%%%%%%%%%%%%%%%%%%%%%%%%%%%%%%%%%%%%%%
\section{Introduction}
%%%%%%%%%%%%%%%%%%%%%%%%%%%%%%%%%%%%%%%%%%%%%%%%%%%%%%%%%%%

\subsection{The background}
Graph exploration by mobile agents is a well-researched task with many practical applications. Communication networks may be explored in order to search for a data item located at an unknown node.
A mobile robot may need to explore a contaminated mine in order to decontaminate it. Often the mobile agent is a simple device, and hence it is natural to model it as a finite automaton. Exploration of subgraphs of the infinite oriented grid $\mathbb{Z} \times \mathbb{Z}$ is of particular interest, as it has an immediate application in exploring terrains in the plane. Indeed, an automaton that can explore a subgraph $G$ of the grid embedded in the plane, with unit distance between adjacent nodes, will see all points of the planar terrain spanned by $G$ if it is equipped with vision of radius 1.

\subsection{The model and the problem}

A mobile agent, modeled as a deterministic finite automaton, navigates in the infinite anonymous oriented grid $\mathbb{Z} \times \mathbb{Z}$. Nodes of the grid do not have labels but ports at each node are labeled $N,E,S,W$, according to the orientation. The agent has a compass that enables it to see the ports $N,E,S,W$ at every visited node. Directions corresponding to ports $N,E,S,W$ are called, respectively, North, East, South, West.
Lines East-West of the grid are called {\em horizontal} and lines North-South are called {\em vertical}.

We are given two half-lines $H_1$ and $H_2$ in the plane,  with common origin $O$.
More precisely, the input to the problem are two unit vectors with real coordinates, corresponding to these half-lines.
Let $W(H_1,H_2)$ denote the (closed) subset of the plane consisting of all points of the plane between $H_1$ and $H_2$, clockwise from $H_1$ to $H_2$, with these half-lines included.
The {\em wedge} corresponding to these half-lines is defined as the subgraph of the grid spanned by all nodes of the grid in $W(H_1,H_2)$. 
The lines $H_1$ and $H_2$ are called {\em boundaries} of the wedge.
If the clockwise angle from $H_1$ to $H_2$ is $0<\alpha \leq 2\pi$, the wedge is called an $\alpha$-wedge.
$\alpha$-wedges for $\alpha<\pi$ are called {\em small} and $\alpha$-wedges for $\alpha\geq \pi$ are called {\em large}.
For given (unit vectors of) half-lines $H_1$ and $H_2$ we want to construct a finite automaton exploring the wedge corresponding to these half-lines. This means that every node of the wedge must be eventually visited by the agent.

The agent starts at some node of the wedge, chosen by an adversary. It may move from a node of the grid to an adjacent node, but some ports at any node may be {\em blocked}: this means that the agent cannot choose this particular port. Non-blocked ports are called {\em free}.
Consider a wedge $W$ corresponding to half-lines $H_1$ and $H_2$. The half-line $H_i$, for $i=1,2$, is called a {\em wall}, if for every edge $e=\{u,v\}$ of the grid, such that $u $ is a node of $W$, $v$ is not a node of $W$, and $H_i$ intersects the closed segment $[u,v]$, the ports at $u$ and at $v$ corresponding to edge $e$ are blocked. Intuitively, the agent cannot cross a wall to leave or re-enter the wedge.
Only a wall can block the agent.
The half-line $H_i$, for $i=1,2$, is called {\em free}, if for every edge $e=\{u,v\}$ of the grid, such that $u $ is a node of $W$, $v$ is not a node of $W$, and $H_i$ intersects the closed segment $[u,v]$, the ports at $u$ and at $v$ corresponding to edge $e$ are free. Intuitively, the agent can cross a free boundary half-line to leave or re-enter the wedge.
A free boundary of a wedge only limits the wedge but does not restrict the moves of the agent.
Entering any node, the agent learns which ports are blocked and which are free.

We consider three types of wedges: {\em walled wedges}, both of whose half-lines are walls, {\em semi-walled wedges}, for which the half-line $H_1$ is a wall and the half-line $H_2$ is free, and {\em free wedges}, both of whose half-lines are  free. In particular, if the wedge is free, this means that the agent has to explore this wedge but navigates in the grid with all ports at all nodes free.
Since in the case of free wedges the agent cannot distinguish the common origin $O$ of its boundaries from any other point of the plane, in the task of exploration of a free wedge we assume that the origin $O$ is at the starting node of the agent. For all other types of wedges, the point $O$ can be arbitrary, and the only restriction is that the agent starts at a (adversarially chosen) node of the grid inside the wedge.

We will see that the type of wedge to be explored, as well as the size of the angle between its boundaries, determine the feasibility of its exploration by an automaton. We will also see that some wedges are impossible to explore by an automaton that cannot mark visited nodes. Hence we allow the agent to use {\em movable pebbles}, and the goal is to explore a given wedge by an automaton using as few pebbles as possible. The agent has $p$ pebbles available to it and it starts carrying all of them. When entering a node, it sees if there is a pebble at this node or not. It may drop a pebble if there is no pebble at the current node and the agent carries a pebble, it may pick a pebble, if there is a pebble at the current node, or it may perform neither of these two actions. We now give a formalization of the above intuitive description.

Let $Ports=\{N,E,S,W\}$, $Free=2^{Ports}$, and $P=\{0,1,\dots,p\}$.
The mobile agent is formalized as a finite deterministic Mealy automaton
${\cal A}=(X,Y,Q,\delta,\lambda,S_0)$. 

 $X=Free \times \{e,f\} \times P$ is the input alphabet, 
$Y=\{N,E,S,W\} \times \{e,f\} \times P$ is the output alphabet. $Q$ is a finite set of states
with a special state $S_0$ called initial.
$\delta:Q\times X \to Q$ is the state transition function, and $\lambda:Q \times X \to
Y$ is the output function.

The meaning of the input and output symbols is the following.  At each step of its functioning, the agent is at some node of the grid.
It sees the set $F\in Free $ of free ports at the current node,
and carries some number $x \in P$ of pebbles. Moreover, the current node is either empty (no pebble), denoted by $e$, or full (has a pebble) denoted by $f$.
The input $I\in X$ gives the automaton information about these facts. The agent is in some state $S$.
 Given the state $S$ and the input $I$,
the agent outputs the symbol $\lambda(S,I)\in \{N,E,S,W\} \times \{e,f\} \times P$ with the following meaning. The first term indicates the port through which the agent moves. The second term determines whether the agent leaves the current node empty or  full, and the third term indicates with how many pebbles the agent transits to the adjacent node. Since the agent must choose a free port at each node, and can only either leave the current node intact and not change the number of carried pebbles, or pick a pebble from a full node leaving it empty (in this case the number of carried pebbles increases by 1), or drop a pebble on an empty node leaving it full, if it carried at least one pebble (in this case the number of carried pebbles decreases by 1), we have the following restrictions on the possible values of the output function $\lambda$:

$\lambda (S,F,\cdot,\cdot)$ must be $(d, \cdot,\cdot)$ with $d\in F$.

$\lambda(S,e,0)$ must be $(\cdot, e,0)$

$\lambda(S,e,z)$ must be either $(\cdot,e,z)$ or $(\cdot,f,z-1)$

 $\lambda(S,f,z)$ must be either $(\cdot,f,z)$ or $(\cdot,e,z+1)$

 Seeing the input symbol $I$ and being in a current state $S$, the agent makes the changes indicated by the output function (it possibly changes the filling of the current node as indicated, possibly changes the number of carried pebbles as indicated, and goes to the indicated
 adjacent node), and transits to state $\delta(S,I)$.
 The agent starts with $p$ pebbles at an empty node in the initial state $S_0$ (hence its initial input symbol is $(\cdot,e,p)$).

\subsection{Our results}

For every pair of half-lines $H_1$ and $H_2$ with common origin $O$, we determine the minimum number $p$ of pebbles sufficient to explore the corresponding wedge $W$
by a finite deterministic automaton, depending on the type and size of the wedge.
We do this by showing a finite deterministic automaton that explores this wedge with $p$ pebbles (depending on the type and size of the wedge) and proving that $p-1$ pebbles are not enough for this task.

The difficulties of the positive results are mostly geometric. Since the agent can make only vertical and horizontal moves, and the slopes of the wedge boundaries can be arbitrary, visiting all grid nodes of the wedge often requires complicated navigating techniques to avoid missing nodes close to the walls which the agent cannot cross. These techniques that guarantee exploration of a given wedge, regardless of the initial position of the agent in the wedge, are the main methodological novelty of the paper. Such difficulties did not exist in the exploration of the obstacle-free grid considered in \cite{ELSUW}.  On the other hand, the difficulties of the negative results are due to considerations regarding bounded memory of the agent and the optimal number of pebbles: it has  to be proved that an insufficient number of pebbles does not permit an agent modeled by an automaton to ``remember'' sufficiently much information, and consequently the agent must enter a loop of states and miss some nodes of the wedge.

We show that for small wedges, the minimum number of pebbles is (cf. Fig. \ref{table}):
\begin{itemize}
\item
0, if the wedge is walled,
\item
1, if the wedge is semi-walled,
\item
2, if the wedge is free.
\end{itemize}

Then we show that for large wedges, the minimum number of pebbles is:
\begin{itemize}
\item
1, if the wedge is walled,
\item
2, if the wedge is semi-walled,
\item
3, if the wedge is free.
\end{itemize}

\begin{figure}[h]
   \vspace*{-.194in}
     \centering
    \includegraphics[scale = .8]{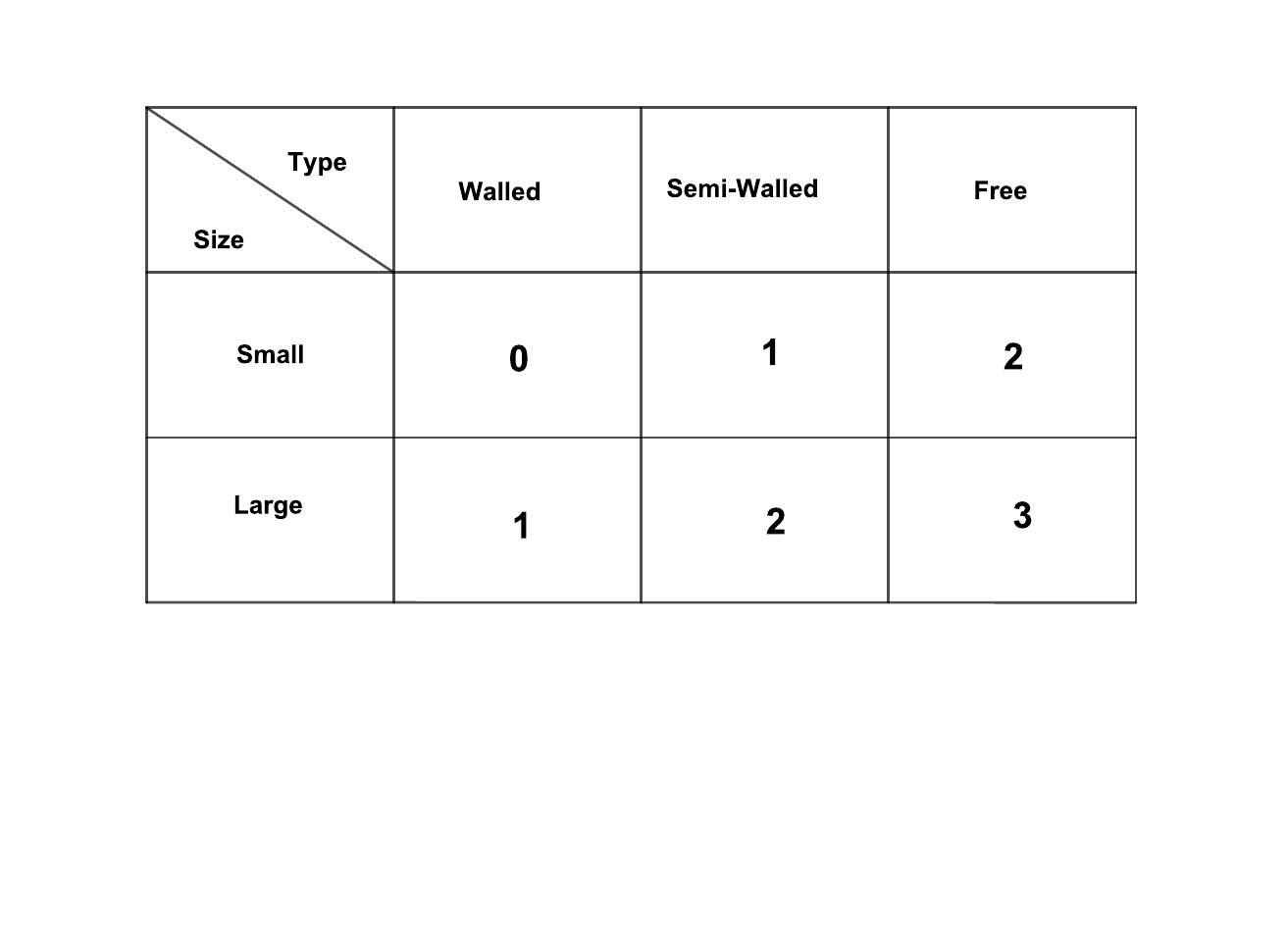}
   \vspace*{-.8in}
    \caption{Summary of our results: in each box, the minimum number of pebbles sufficient to explore a wedge depending on its type and size}
    \label{table}
\end{figure}

We observe that a small wedge is not always a connected subgraph of the grid, hence the first of the six above statements should be read as follows:  every {\em connected} small wedge can be explored by an automaton with no pebbles. In our solution, the automaton explores the connected component of its starting node in case when the wedge is disconnected. (Clearly, disconnected walled wedges cannot be entirely explored, regardless of the number of pebbles). For small semi-walled or free wedges the agent may cross the free boundary half-line (going beyond the wedge) and the corresponding results hold without any restriction: the wedge can always be explored with the stated number of pebbles. Large wedges are always connected subgraphs of the grid, hence there is no restriction either.

The last of our six results says that the minimum number of pebbles to explore a large free wedge is three. Its positive part, i.e., showing how to explore any free wedge using an automaton with three pebbles, is a strengthening
of the result from \cite{ELSUW} which says that 4 (semi-synchronous) deterministic automata can explore the grid. Indeed, three of the automata can simulate pebbles, and the entire grid is a large free wedge. 

%The negative part of our last result follows  from the result of \cite{BUW} which says that 3 (semi-synchronous) deterministic automata cannot explore the whole grid. In fact the proof from \cite{BUW} shows more:  3 (semi-synchronous) deterministic automata cannot explore even all grid nodes in a half-plane, i.e., (in our terminology) cannot explore any large free wedge. Since automata can simulate pebbles, this proves that an automaton with two pebbles cannot explore a large free wedge.

Our approach should be contrasted with that of the classic papers \cite{BK78,B78,H81} that dealt with exploring subgraphs of the grid by agents with pebbles. While in these papers the explored subgraphs were {\em finite} and the task was to construct an automaton exploring all subgraphs from a given class, with as few pebbles as possible, we consider a given {\em infinite} subgraph of the grid (a wedge) and want to construct an automaton with as few pebbles as possible that explores this given subgraph. Of course, for any given {\em finite} graph it is trivially possible to explore it by a (sufficiently large) automaton without pebbles.

A finite deterministic automaton can remember a finite number of bits by encoding them in its states. Following the practice in the literature on automata navigating in graphs, and in order to increase readability, we present our positive results by designing exploration algorithms that need only remember a constant number of bits and thus can be executed by such automata, rather than formally constructing an automaton by defining its output and state transition functions.

%\noindent
%{\bf The model.}
% We consider the infinite oriented grid $\mathbb{Z} \times \mathbb{Z}$ represented as the set of unit square cells tiling the two-dimensional plane,
%with all cell sides vertical or horizontal. Each cell has 4 adjacent cells, North, East, South and West of it. Some cells of the grid contain a brick, i.e., 
% are {\em full}, other cells are {\em empty}. The subgraph of the grid induced by full cells is initially connected. 
% At each step of the algorithm this subgraph can change, due to the actions of the robot, described below. At each step, the subgraph induced by the full cells is called the current {\em field}.  Any maximal connected  subgraph of the current field is called a {\em component}.
% Throughout the paper, the {\em distance} between two cells $(x,y)$ and $(x',y')$  of the grid is the Manhattan distance between them, i.e., $|x-x'|+|y-y'|$.
% The number of cells of a field is called its {\em size}, and the distance between two farthest cells of a field is called its {\em span}. 
% A nest of size $z$ is a field that has the minimum span among all fields of size $z$.

\subsection{Related work}

Exploration of various environments by mobile agents
has been studied for many years
(cf.~\cite{H89,RKSI}). The literature of this domain can be divided into two parts, according to the environment where the agents operate:
it can be either a geometric terrain, possibly with obstacles, or a network modeled as a graph in which the agents move along edges.

In the plane, a related problem is that of pattern formation \cite{DFSY,DPV,SY}. Robots, modeled as points moving in the plane have to form a pattern given as input. This task has been mostly studied in the context of  asynchronous oblivious robots having full visibility of other robots positions.

The graph setting can be further specified in two different
ways. In \cite{AH,BFRSV,BS,DP,FI04} the mobile agent explores strongly
connected directed graphs and it can move only from tail to head
 of a directed edge, not vice-versa. In
\cite{ABRS,BRS2,B78,DFKP,DJMW,DKK,PaPe,R79-80} the explored
graph is undirected and the agent can traverse edges in both
directions. Graph exploration scenarios can be also classified from a different point of view. 
It is either assumed that nodes of the
graph have unique labels which the agent can recognize (as in,
e.g.,~\cite{DP,DKK,PaPe}), or it is assumed that nodes are anonymous
(as in, e.g.,~\cite{BFRSV,BS,B78,CDK,R79-80}). In our case, we work with wedges of  the infinite anonymous grid, hence it is
an undirected anonymous graph scenario.
Two main efficiency measures are adopted in papers dealing with
graph exploration. One is the completion time of this task, measured
by the number of edge traversals, (cf., e.g.,~\cite{PaPe}), and the other is the
memory size of the agent, measured either in bits or by the number of
states of the finite automaton modeling the agent (cf.,
e.g., \cite{DFKP,FIPPP,FI04}). In the present paper we are not concerned with minimizing the memory size but we assume
that this memory is finite. However, we want to minimize the number of pebbles used by the agent.

The capability of an agent to explore anonymous undirected graphs has
been studied in, e.g., \cite{BK78,B78,DFKP,FIPPP,K79,R79-80}. 
In this context, the explored graphs were finite and the task was to construct a single automaton that explores a given class of these graphs.
In particular, it was shown in \cite{R79-80} that no finite automaton can
explore all cubic planar graphs (in fact no finite set of finite
automata can cooperatively perform this task). Budach \cite{B78} proved that a single automaton cannot explore all mazes, i.e., finite connected subgraphs of the infinite grid.
Hoffmann \cite{H81} proved that one pebble does not help to do it. By contrast, Blum and Kozen \cite{BK78} showed that this task can be accomplished by two cooperating automata or by a single automaton with two pebbles.
More recently, it was shown in \cite{DHK} that an agent with $\Theta(\log\log n)$ distinguishable pebbles and bits of memory can explore any bounded-degree graph with at most $n$ nodes, and that these bounds are tight. 
The size of port-labeled graphs which cannot be explored by a given automaton was
investigated in \cite{FIPPP}.

Recently many authors studied the problem of exploring the infinite anonymous oriented grid by cooperating agents modeled as either deterministic or probabilistic automata. Such agents are sometimes called ants. It was shown in \cite{ELSUW} that 3 randomized or 4 deterministic automata can accomplish this task. Then matching lower bounds were proved: the authors of \cite{CELU} showed that 2 randomized automata are not enough for exploring the grid, and the authors of \cite{BUW} proved that 3 deterministic automata are not enough for this task. On the other hand, an old result from \cite{BS77} shows that an automaton with 7 pebbles can explore any connected subgraph of the infinite grid.
 
\section{Small wedges}

\subsection{Walled wedges}

In this section we show that for any small connected walled wedge there exists an automaton that explores it without any pebbles.
We present two exploration algorithms: as a warm-up we formulate a simpler algorithm that works for {\em acute} wedges ($\alpha$-wedges with $\alpha<\pi/2$) and then a more difficult algorithm that works for  {\em obtuse} wedges ($\alpha$-wedges with $\pi/2 \leq \alpha<\pi$).

\subsubsection{Acute wedges}

We start with the observation that an acute wedge may be disconnected (see Fig. \ref{fig-1}). Disconnected walled wedges cannot be entirely explored by an agent, regardless of the number of pebbles.
Hence our algorithm will entirely explore only connected acute walled wedges.

\begin{figure}[h]
   \vspace*{-.194in}
     \centering
    \includegraphics[scale = .6]{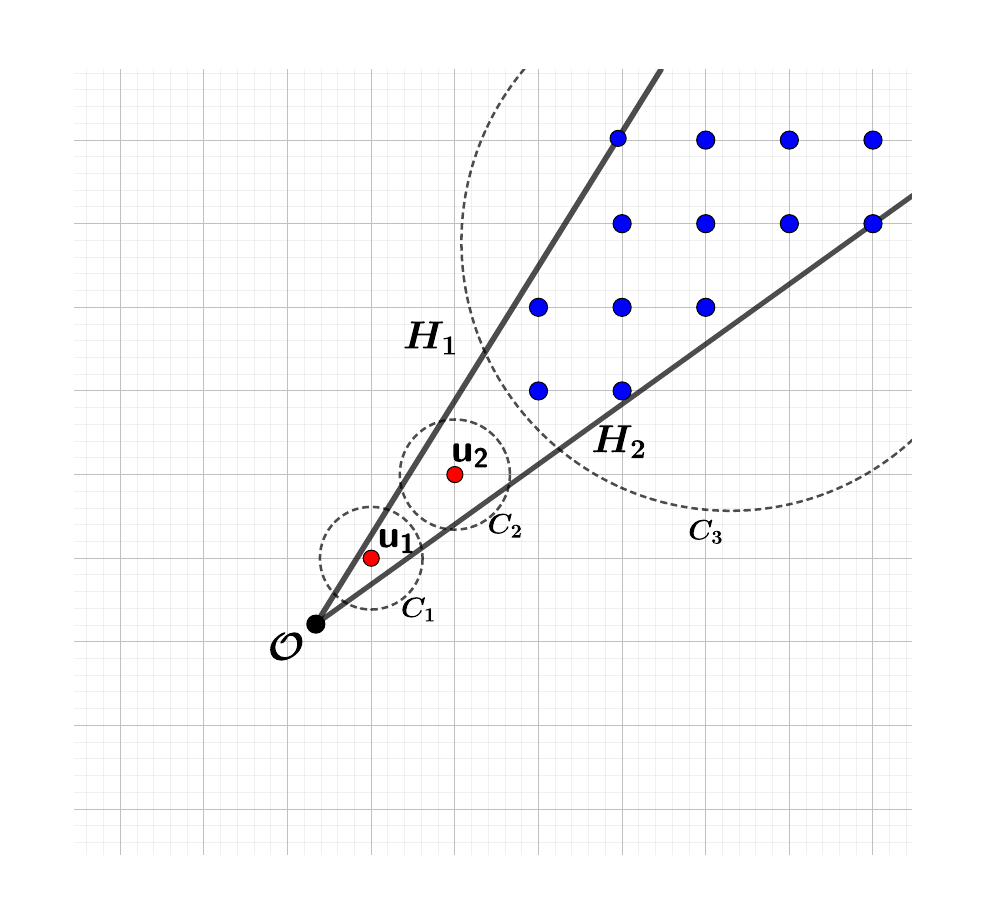}
   \vspace*{-.2in}
    \caption{A disconnected acute wedge with three connected components $C_1$, $C_2$ and $C_3$}
    \label{fig-1}
\end{figure}

Let $W$ be an acute wedge corresponding to half-lines $H_1$ and $H_2$ with common origin $O$. Observe that, regardless of the position of $H_1$ and $H_2$, both these half-lines can be intersected either by a horizontal or by a vertical line of the grid. Indeed, if both boundaries of $W$ are either North or both are South of the origin $O$ then both of them are intersected by a horizontal line, otherwise  both of them are intersected by a vertical line. Without loss of generality, assume that both boundaries are North of the origin $O$, and hence can be intersected by a horizontal line. The algorithm for the other cases is similar.
\begin{figure}[h]
   \vspace*{-.194in}
     \centering
    \includegraphics[scale = .5]{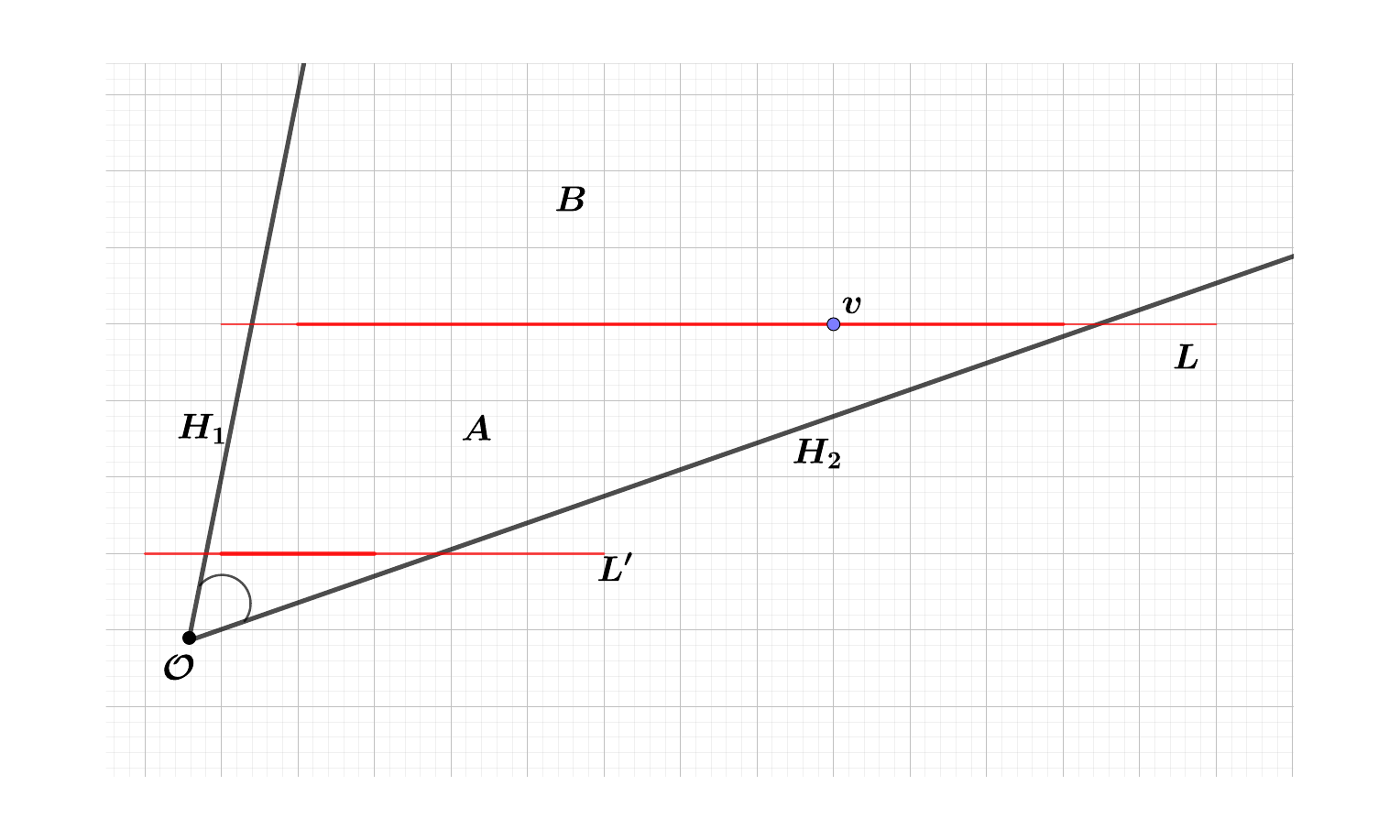}
   \vspace*{-.2in}
    \caption{An illustration of the Algorithm {\tt Explore Acute Walled}}
    \label{fig-2}
\end{figure}

The high-level idea of the algorithm is the following. The agent starts at some node $v$ of the wedge and explores the segment of the horizontal line $L$ containing $v$,  that is contained in the wedge (see Fig. \ref{fig-2}). It does this by first going West until finding a blocked port $W$ and then going East until finding a blocked port $E$. Then it explores the (finite) part $A$ of the wedge South of $L$ in phases. In each phase it explores the  segment of the current horizontal line included in the wedge, and then transits to the horizontal line immediately South from the current one. This is done until there is no free port $S$ from any node of the currently visited segment. At this point the agent is at a horizontal line $L'$ and starts the trip North, again in phases. In each phase it explores the segment of the current horizontal line, contained in the wedge, as before. In the next phase it transits to the horizontal line immediately North of it, explores the segment and so on, indefinitely. Notice that at the beginning of its trip North, the agent re-explores already explored segments in the part $A$ of the wedge, but it has to do these repetitions, as it cannot remember the number of segments previously explored on its way South of $L$.

We will use the following procedures. The first one explores the current horizontal segment and checks if some port $S$ at a node of the explored segment is free.

\vspace*{0.5cm}
{\bf Procedure} {\tt Explore Horizontal}

\hspace*{1cm}$open-south :=false$

\hspace*{1cm}{\bf repeat}

\hspace*{2cm}take port $W$ 

\hspace*{1cm}{\bf until} port $W$ is free;

\hspace*{1cm}{\bf repeat}

\hspace*{2cm}take port $E$ 

\hspace*{1cm}{\bf until} port $E$ is free;

\hspace*{1cm}{\bf if} the port $S$ at some node is free, {\bf then} $open-south :=true$

\vspace*{0.5cm}

The second procedure executes the trip South of $L$.

\vspace*{0.5cm}
{\bf Procedure} {\tt Explore South}

\hspace*{1cm}{\tt Explore Horizontal}

\hspace*{1cm}{\bf while}  $open-south =true$ {\bf do}

\hspace*{2cm}{\bf while} port $S$ blocked take port $W$;

\hspace*{2cm}take port $S$

\hspace*{2cm}{\tt Explore Horizontal}

\vspace*{0.5cm}

The third procedure executes the trip North, starting at the southern-most segment explored by procedure  {\tt Explore South}. Now there is no need to check if it is possible to go North from the currently explored segment because it is always the case. The trip North is executed indefinitely, covering first the finite part $A$ and then the infinite part $B$.

\vspace*{0.5cm}
{\bf Procedure} {\tt Explore North}

\hspace*{1cm} {\bf repeat forever}

\hspace*{2cm}{\tt Explore Horizontal}

\hspace*{2cm}{\bf while} port $N$ blocked {\bf do} 

\hspace*{3cm}take port $W$;

\hspace*{2cm}take port $N$

\vspace*{0.5cm}

Now the exploration algorithm for acute walled wedges can be succinctly formulated as follows.

\vspace*{0.5cm}

{\bf Algorithm} {\tt Explore Acute Walled}

\hspace*{1cm}{\tt Explore South}

\hspace*{1cm}{\tt Explore North}

\begin{proposition}\label{acute}
For every connected acute walled wedge there exists an automaton that explores it without pebbles.
\end{proposition}

\begin{proof}
Consider a connected acute walled wedge $W$.
Without loss of generality, assume the case considered for algorithm {\tt Explore Acute Walled}, i.e., that both boundaries of $W$ are North of the origin $O$. All other cases are similar. 
Upon completion of procedure {\tt Explore South} the agent explores the part of the wedge consisting of all segments of horizontal lines South of the horizontal line where the agent starts.
Then, at each turn of the loop ``repeat forever'' of procedure {\tt Explore North}, the agent explores consecutive horizontal segments of the wedge starting at the southern-most one. Hence all nodes of the wedge are eventually explored. Observe that both procedures {\tt Explore South} and {\tt Explore North} can be executed by a finite automaton independent of the starting node.
Note that, if the wedge were not connected, then the agent starting in the infinite connected component would miss the finite southern-most part of the wedge disconnected from the component containing its starting node.
\end{proof}

\subsubsection{Obtuse small wedges}

We now consider an obtuse small walled wedge, i.e., an $\alpha$-wedge for $\pi/2 \leq \alpha<\pi$. Notice that all such wedges are connected. 
We define a line to be  {\em rational} if it contains at least two distinct nodes of the grid. If a line is rational then it is either vertical or its slope (the tangent of its angle with a horizontal line) is rational.
For positive integers $x$ and $y$, a rational line is called a $(x,y)$-line if the slope of this line is $y/x$.
For a small wedge with boundaries $H_1$ and $H_2$, call a line {\em cutting} if it intersects both boundaries $H_1$ and $H_2$ of the wedge 
%in points $P_1$ and $P_2$, such that the segment $P_1P_2$ is a subset of the convex hull of $H_1 \cup H_2$ 
(see Fig. \ref{fig-3}).

\begin{figure}[h]
   \vspace*{-.194in}
     \centering
    \includegraphics[scale = .5]{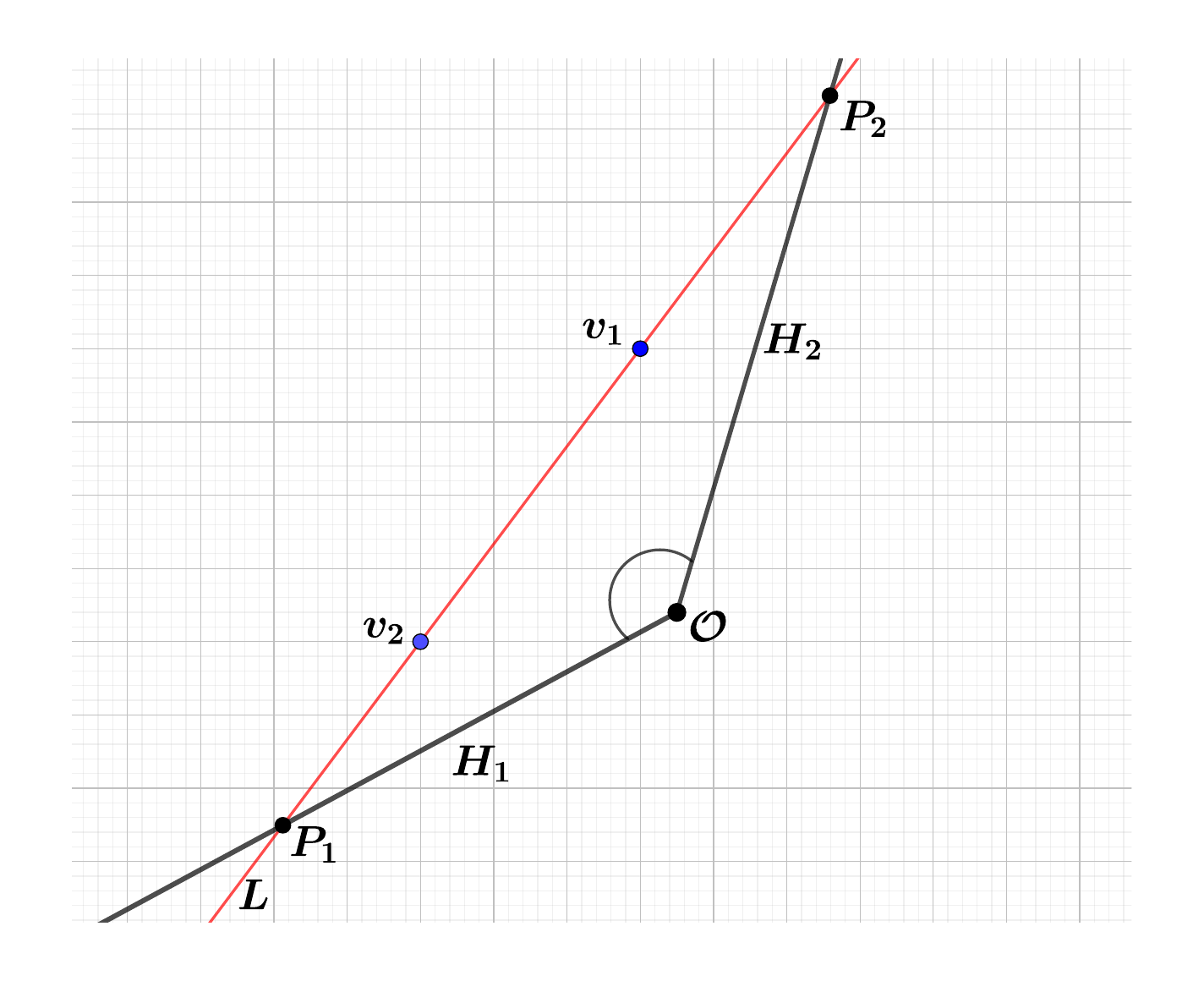}
   \vspace*{-.2in}
    \caption{A cutting rational line $L$: $v_1$ and $v_2$ are grid nodes in $L$}
    \label{fig-3}
\end{figure}

We will use the following simple lemma.

\begin{lemma}\label{cutting}
For any small wedge, there exists a rational cutting line.
\end{lemma}

\begin{proof}
Since the angle between the boundaries of the wedge is less than $\pi$, there exists a cutting line $L'$. Let $\beta'$\ be the acute angle of $L'$ with a horizontal line. Since the tangent function is continuous, there exists a cutting line $L''$ (which is a slight rotation of $L'$) such that 
 the tangent of its angle $\beta''$ with a horizontal line is rational. Now, the line $L''$ can be shifted by a parallel translation to get a rational cutting line $L$. 
\end{proof}

%Let $L$ be a rational cutting line of the wedge passing through the starting node $v$ of the agent.
%Suppose that the line $L$ passes through nodes $(a,b)$ and $(a+x,b+y)$. Hence $L$ is a $(x,y)$-line (see Fig. \ref{fig-4}). 

Consider a wedge with boundaries $H_1$ and $H_2$. Without loss of generality assume that the vector determining $H_1$ has both components negative and the vector determining $H_2$ has both components positive. The algorithm for other cases is similar (including the limit case when one of the boundaries is either horizontal or vertical).

Given any node $v=(a,b)$ of a $(x,y)$-line $L$,  we define the {\em staircase} of $L$ containing $v$ as the polygonal line $(\dots (a-2x,b-2y),(a-2x,b-y),(a-x,b-y),(a-x,b),(a,b),(a,b+y),(a+x,b+y),(a+x,b+2y,(a+2x,b+2y),\dots)$. Call nodes of the staircase situated on the line $L$ (i.e., nodes $\dots ,(a-2x,b-2y),(a-x,b-y),(a,b),(a+x,b+y),(a+2x,b+2y),\dots$)  {\em cardinal nodes} of this staircase (see Fig. \ref{fig-4}).

\begin{figure}[h]
   \vspace*{-.194in}
     \centering
    \includegraphics[scale = .6]{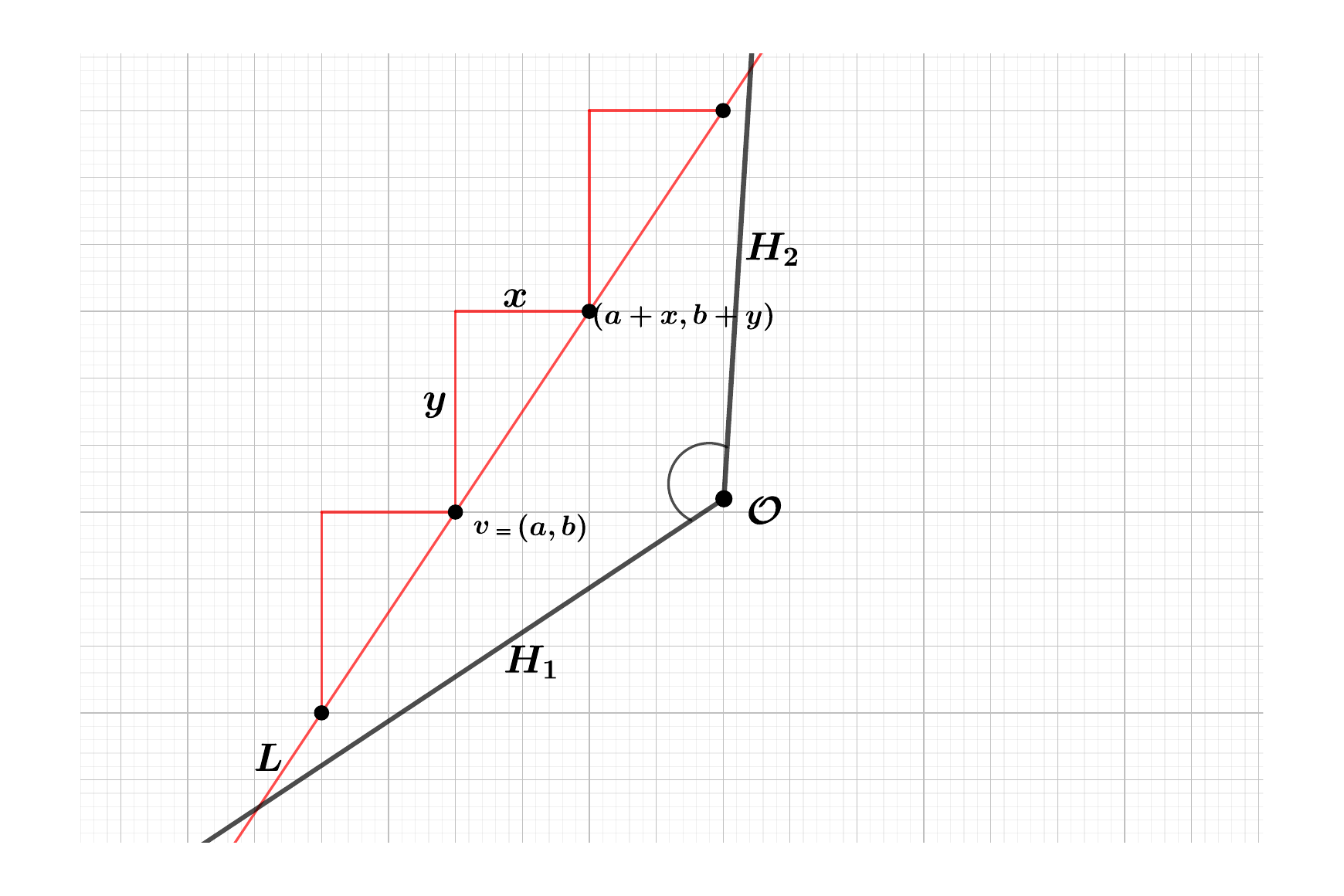}
   \vspace*{-.2in}
    \caption{An $(x,y)$-line and a staircase for $x=2$ and $y=3$. Cardinal nodes are marked.}
    \label{fig-4}
\end{figure}

For any $(x,y)$-line $L'$ parallel to $L$ and any node $w=(a',b')$ of the grid contained in $L'$, we define the {\em box} of $w$, as the set of all nodes $(a'-x',b'-y')$ of the grid, such that $0\leq x'\leq x$ and $0 \leq y' \leq y$. Exploring a box means visiting all nodes of it.
The {\em chain of boxes} for $L'$ and $w$ is the set of boxes for nodes $(\dots, a'-2x,b'-2y),(a'-x,b'-y), (a',b'), (a'+x,b'+y), (a'+2x,b'+2y),\dots)$ (see Fig. \ref{boxes}).

\begin{figure}[h]
   \vspace*{-.194in}
     \centering
    \includegraphics[scale = .5]{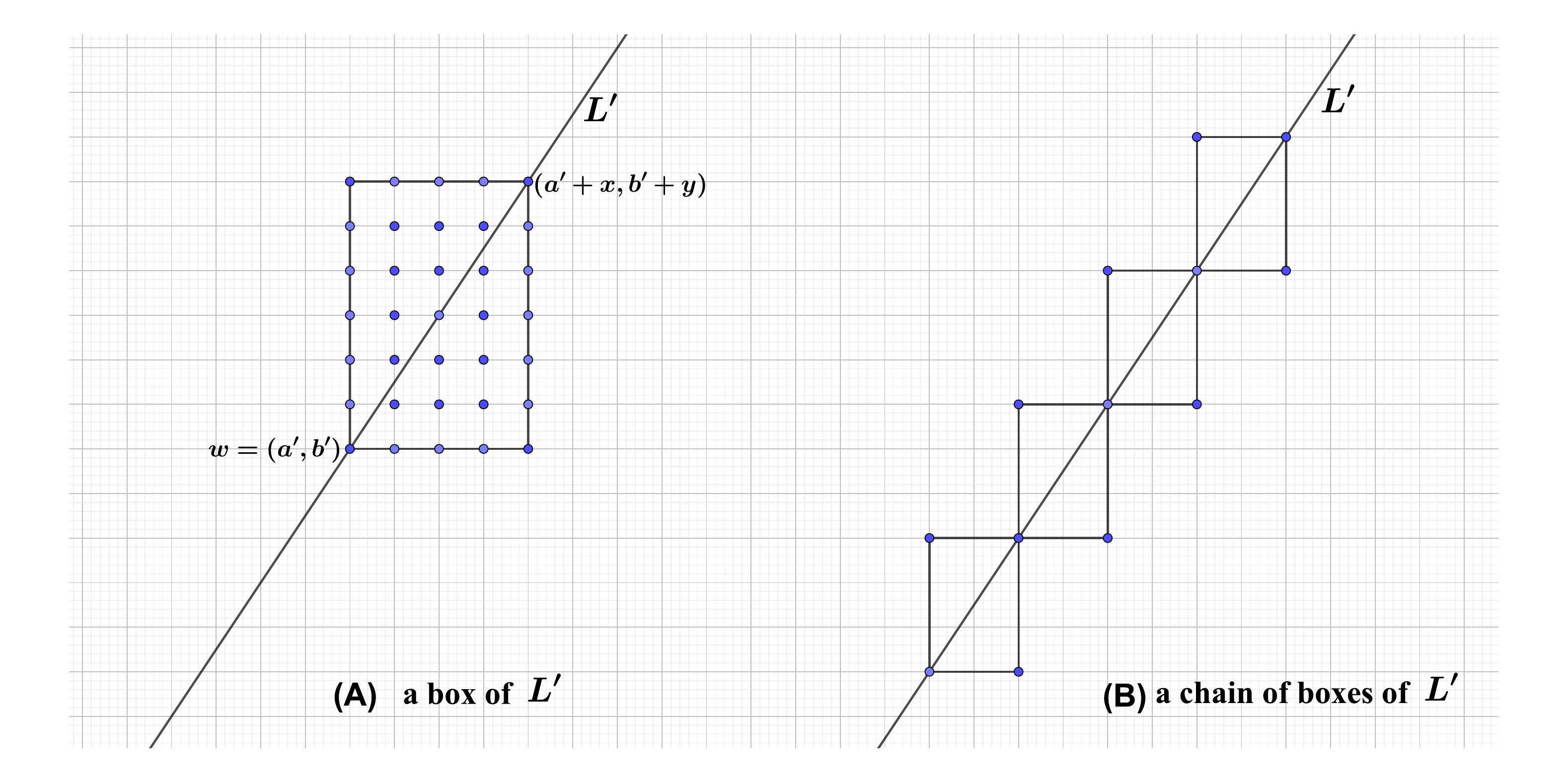}
   \vspace*{-.2in}
    \caption{A box and a chain of boxes}
    \label{boxes}
\end{figure}

The high-level idea of the algorithm is an extension of the idea of Algorithm {\tt Explore Acute Walled} but it is significantly more complicated. The agent starts at some node $v$ of the wedge. 
Let $L$ be a rational cutting line of the wedge passing through node $v$.
%Suppose that the line $L$ passes through nodes $(a,b)$ and $(a+x,b+y)$. Hence $L$ is a $(x,y)$-line. 

If the agent could travel along lines parallel to $L$, we could use the same method as in Algorithm {\tt Explore Acute Walled}: first explore all segments of the wedge South of $L$ and parallel to it, and then go North exploring consecutive segments parallel to $L$. Since this is impossible, we proceed as follows. First we use staircases of lines parallel to $L$ to go down to the {\em bottom line} $L^*$ parallel to $L$, which is close to the origin $O$ and will be precisely defined later. Then we explore all grid points of the triangle bounded by the two boundaries and by the bottom line. After having done this, we iteratively explore consecutive parts of the wedge bounded by the boundaries and consecutive lines $L_i$ parallel to $L^*$, in each iteration going one step North. We want to keep the invariant that after the $i$-th iteration all nodes of the wedge below line $L_i$ are explored. In order to do so, we explore consecutive chains of boxes of these lines. This is complicated by the fact that some  portions of such triangles close to the intersection points of $L_i$ with the boundaries could be omitted by the exploration of the chain of boxes. This is why we have additional procedures to cover these pieces.

%Similarly as before, we first  explore all such segments South of $L$ and then go North exploring segments of consecutive lines parallel to $L$. The additional difficulty is to guarantee that at the time when the trip North starts at some segment of a staircase, all nodes of the grid South of it are already explored.

We now proceed with the detailed description of the algorithm. The aim of its first part is to get to the bottom line $L^*$ which is so close to the origin $O$ that the number of nodes of the wedge inside the triangle below it is bounded, regardless of the starting node $v$.
In order to achieve this we will use the following procedures. 

%The first one explores the fragment of the current staircase South of the starting (cardinal) node at this staircase and ends at the southern-most cardinal node of its segment contained in the wedge.
%
%\vspace*{0.5cm}
%
%
%{\bf Procedure} {\tt Explore Staircase Down}
%
%
%$count :=y$
%
%\hspace*{1cm}{\bf while}  $count=y$ {\bf do}
%
%\hspace*{2cm} {\bf for} $i:=1$ {\bf to} $x$ {\bf do}
%
%\hspace*{3cm} take port $W$;
%
%\hspace*{2cm} $count :=0$
%
%\hspace*{2cm} {\bf while} port $S$ free and $count <y$ {\bf do}
%
%\hspace*{3cm} take port $S$;
%
%\hspace*{3cm} $count := count +1$;
%
%\hspace*{1cm}{\bf for }  $i:=0$ {\bf  to } $count -1$ {\bf do}
%
%\hspace*{2cm} take port $N$;
%
%\hspace*{1cm}{\bf for }  $i:=1$ {\bf  to } $x$ {\bf do}
%
%\hspace*{2cm} take port $E$;
%

The first procedure starts at a cardinal node of the current staircase segment, explores the upper part of it contained in the wedge, above this node, and ends at the northern-most cardinal node of
this segment. It also explores the part of the northern-most unfinished box included in the wedge.

\vspace*{.5cm}

{\bf Procedure} {\tt Explore Staircase Up}

$count :=x$

\hspace*{1cm}{\bf while}  $count=x$ {\bf do}

\hspace*{2cm} {\bf for} $i:=1$ {\bf to} $y$ {\bf do}

\hspace*{3cm} take port $N$;

\hspace*{2cm} $count :=0$

\hspace*{2cm}{\bf while} port $E$ free and $count <x$ {\bf do}

\hspace*{3cm} take port $E$;

\hspace*{3cm} $count := count +1$;

\hspace*{1cm}{\bf for }  $i:=0$ {\bf  to } $count -1$ {\bf do}

\hspace*{2cm} take port $W$;

\hspace*{1cm}{\bf for }  $i:=1$ {\bf  to } $y$ {\bf do}

\hspace*{2cm} take port $S$;

\hspace*{2cm} $count :=0$;

\hspace*{2cm} {\bf while} port $E$ free and $count<x$ {\bf do}

\hspace*{3cm} take port $E$;

\hspace*{3cm} $count := count +1$;

\hspace*{2cm} {\bf while} $count>0$ {\bf do}

\hspace*{3cm} take port $W$;

\hspace*{3cm} $count := count -1$;

\vspace*{1cm}

The next procedure consists of exploring the segment of the current staircase, trying to move one step East from some cardinal node of it, then exploring the segment of the staircase of the next parallel line to $L$ and so on. This is done until no such move East is possible from any cardinal point of the current staircase segment. Then the agent returns to the most recent cardinal node of the current segment which happens to be the southern-most cardinal node of~it.

\vspace*{0.5cm}

{\bf Procedure} {\tt Explore Down}

\hspace*{1cm} $done:=false$;

\hspace*{1cm} {\bf while} $done=false$ {\bf do}

\hspace*{2cm} {\tt Explore Staircase Up};

\hspace*{2cm} $flag:=false$; $count:=y$;

\hspace*{2cm} {\bf while}  $count=y$ and $flag=false$ {\bf do}

\hspace*{3cm} {\bf for} $i:=1$ {\bf to} $x$ {\bf do}

\hspace*{4cm} take port $W$;

\hspace*{3cm} $count :=0$

\hspace*{3cm} {\bf while} port $S$ free and $count <y$ {\bf do}

\hspace*{4cm} take port $S$;

\hspace*{4cm} $count := count +1$;

\hspace*{3cm} {\bf if} $count=y$ {\bf then}

\hspace*{4cm} {\bf if} port $E$ free {\bf then} 

\hspace*{5cm} take port $E$; $flag :=true$

\hspace*{3cm} {\bf else}

\hspace*{4cm} $done:=true$;

\hspace*{1cm}{\bf for }  $i:=0$ {\bf  to } $count -1$ {\bf do}

\hspace*{2cm} take port $N$;

\hspace*{1cm}{\bf for }  $i:=1$ {\bf  to } $x$ {\bf do}

\hspace*{2cm} take port $E$;

\vspace*{0.5cm}

Let $L^*$ be the line parallel to $L$ at which procedure  {\tt Explore Down} stops. Call it the {\em bottom line}. 
Although, by definition of $L^*$, it is impossible to move East within the wedge from any cardinal node of the staircase of $L^*$, there may still be grid nodes in the triangle bounded by the two boundaries of the wedge and by the line $L^*$. The aim of the next procedure is to visit all these nodes.

First consider the following task. Given a node $A$ of the wedge and a finite set $\Sigma$ of nodes of the wedge, the agent starting at node $A$ has to visit all nodes from the set $\Sigma$ and come back to $A$.
The procedure {\tt Visit} $(A, \Sigma)$ accomplishing this task is the following. Let $\Sigma=\{B_1,\dots , B_k\}$.  
Order all paths in the wedge, considered as sequences of  ports $N,E,S,W$ that the agent takes at consecutive steps, in lexicographic order.
For every node $B_i$ let $\pi_i$ be the lexiocographically smallest path  from $A$ to $B_i$ (such a path exists by connectivity of the wedge). 
%and let $rev(\pi_i)$ be the reverse path. 
We define the procedure traverse $\pi_i$ as follows. The agent tries consecutive ports of the path as long as they are free. If a port is blocked, the agent stops.
We define the procedure traverse $rev(\pi_i)$ as traversing the reverse of the part of the path visited during traverse $\pi_i$.

The procedure can be formulated as follows.

\vspace*{0.5cm}

{\bf Procedure} {\tt Visit} $(A, \Sigma)$

\hspace*{1cm} {\bf  for} $i=1$ {\bf to} $k$ {\bf do}

\hspace*{2cm} traverse $\pi_i$; traverse $rev (\pi_i)$

\vspace*{0.5cm}

Now consider a wedge with boundaries $H_1$ and $H_2$ and fix a rational  cutting line $L$ that exists by Lemma \ref{cutting}. Suppose that this is an $(x,y)$-line. For any starting node $v$ of the agent in this wedge there exists a unique bottom line $L^*$. Notice that while the slope of $L$ was chosen independently of any starting node of the agent, the bottom line $L^*$ parallel to $L$ depends on the cardinal nodes in the execution of procedure {\tt Explore Down}, and hence it depends on the starting node $v$. However, for all starting nodes on the same horizontal line of the grid, the bottom line is the same, as 
for all such starting nodes, the cardinal nodes in the execution of procedure {\tt Explore Down} are the same. Also notice that if the vertical distance between two starting nodes is a multiple of $y$, then the bottom line is the same, for the same reason. It follows that if the starting node of the agent is $(a,ty+r)$, where $0\leq r<y$ is fixed, then the bottom line $L^*$ is fixed and the cardinal node at which the agent ends procedure {\tt Explore Down} is fixed. For such a fixed $0\leq r<y$, let $A_r$ be this fixed cardinal node and let $\Sigma_r$ be the finite set of grid nodes in the (closed) triangle $T_r$ bounded by $H_1$, $H_2$ and $L^*$. For any starting node $v$ there is a unique $r$, and hence a unique node $A_r$ and a unique set $\Sigma_r$. Moreover, knowing the slopes of boundaries $H_1$ and $H_2$, integers $x$ and $y$ determining the slope of $L$ can be chosen and all nodes $A_r$ and sets $\Sigma_r$ can be precomputed.  The following procedure visits all nodes from $\Sigma_r$ starting from $A_r$, for any possible $0\leq r<y$. Hence it explores all possible {\em bottom triangles} bounded by the boundaries of the wedge and by all possible bottom lines. We do this because it is impossible to find the exact location of the origin $O$ and hence to determine the exact bottom triangle for a given starting node of the agent.

\vspace*{0.5cm}

{\bf Procedure} {\tt Explore Bottom Triangles}

\hspace*{1cm} {\bf for} $r:=0$ {\bf to} $y-1$ {\bf do}

\hspace*{2cm} {\tt Visit} $(A_r, \Sigma_r)$

\vspace*{0.5cm}

Since for any starting node $v$, some $r$ is the actual one and thus some $A_r$ and $\Sigma_r$ are the actual ones, we have the following lemma.

\begin{lemma}\label{triangle}
Upon completion of procedure {\tt Explore Bottom Triangles}, all grid nodes in the bottom triangle (bounded by the boundaries of the wedge and by the bottom line) are visited and the agent is at a cardinal node of the staircase of the bottom line.
\end{lemma}

We now present several procedures that will be used in formulating the second part of the algorithm.
Our first aim is to explore a chain of boxes. We do this  by formulating 4 procedures.
The first of them starts at a node, explores $x$ nodes West of it and comes back.

\vspace*{0.5cm}

{\bf Procedure} {\tt Horizontal Trip}

\hspace*{1cm} $countW:=0$;

\hspace*{1cm} {\bf while} $countW<x$ {\bf do}

\hspace*{2cm} take port $W$;

\hspace*{2cm} $countW:=countW+1$;

\hspace*{1cm} {\bf while} $countW>0$ {\bf do}

\hspace*{2cm} take port $E$;

\hspace*{2cm} $countW:=countW-1$;

\vspace*{0.5cm}

The next procedure explores the box (or the accessible rectangular part of it) of the node where it starts, i.e., at the beginning of the procedure the agent is at the North-East corner of the box to be explored. At the end, the agent is at the South-West corner of the box or of the accessible rectangular part of it.
The boolean $complete$ becomes $true$ in the case when the procedure ended because the entire box was explored.

\vspace*{0.5cm}

{\bf Procedure} {\tt Explore Box}

\hspace*{1cm}  {\tt Horizontal Trip};

\hspace*{1cm} $countS :=1$;

\hspace*{1cm} $complete:=false$;

\hspace*{1cm} {\bf while} $S$ is free and $countS <y$ {\bf do}

\hspace*{2cm} take port $S$;

\hspace*{2cm}  {\tt Horizontal Trip};

\hspace*{2cm} $countS:=countS+1$;

\hspace*{1cm} {\bf if} $countS=y$ {\bf then}

\hspace*{2cm} $complete:=true$;

\hspace*{1cm} $countW:=0$;

\hspace*{1cm} {\bf while} $countW<x$ {\bf do}

\hspace*{2cm} take port $W$;

\hspace*{1cm} $countW:=countW+1$;

\vspace*{0.5cm}

The next procedure finishes the exploration of a box that was not completely explored by procedure {\tt Explore Box}.
It will be called after an execution of procedure {\tt Explore Box} with $complete=false$ and the initial value of $countS$ in {\tt Finish Box} will come from this execution.
Upon completion of the procedure the agent is at the South-West corner of the current box.

\vspace*{0.5cm}

{\bf Procedure} {\tt Finish Box}

\hspace*{1cm} {\bf while} $countS<y$ {\bf do}

\hspace*{2cm} take port $S$;

\hspace*{2cm}  $countS:=countS+1$;

\hspace*{2cm} $countE:=0$;

\hspace*{2cm} {\bf while} $E$ is free {\bf do}

\hspace*{3cm} take port $E$;

\hspace*{3cm} $countE:=countE+1$;

\hspace*{2cm} {\bf while} $countE>0$ {\bf do}

\hspace*{3cm} take port $W$;

\hspace*{3cm} $countE:=countE-1$;

\vspace*{0.5cm}

The above procedures {\tt Explore Box} and {\tt Finish Box} are now combined in the following procedure {\tt Explore Chain}. This procedure starts at a node $w$ in a line $L'$ parallel to $L$. Node $w$ is the North-East corner of a box $B$. The aim of the procedure is to explore iteratively the part of the chain of boxes of $L'$ starting  at box $B$, and using procedure {\tt Explore Box}   going down along $L'$ towards the wall $H_1$. The last box is explored partially and then the accessible part of it is explored by procedure {\tt Finish Box}. 
\vspace*{0.5cm}

{\bf Procedure} {\tt Explore Chain}

\hspace*{1cm} $complete: = true$;

\hspace*{1cm}  {\bf while} $complete = true$ {\bf do}

\hspace*{2cm} {\tt Explore Box};

\hspace*{1cm} {\tt Finish Box};

\vspace*{0.5cm}

Unfortunately, after the completion of procedure {\tt Explore Chain} there may be still a portion of the triangle between boundaries and the current cutting line that is not covered. This portion is close to the bottom vertex of this triangle and the next procedure aims at exploring it.

Let $P$ be the point where line $L'$ intersects the boundary $H_1$. Let $B$ be the last box (partially) explored by procedure {\tt Explore Chain}.
Let $P_1$ be the South-West corner of $B$ and let
$P_2$ be the point where line $L'$ intersects the southern side of $B$. Such a point exists by the definition of box $B$.
Let $x'$ be the length of the segment $P_1P_2$. Hence $x'\leq x$. Let $a$ be  the vertical distance between the point $P$ and the horizontal line $P_1P_2$.
Let $\gamma$ be the angle between $L'$ and $H_1$ and let $\delta$ be the angle between the vertical line and $L'$ (see Fig. \ref{below}).

\begin{figure}[h]
   \vspace*{-.194in}
     \centering
    \includegraphics[scale = .50]{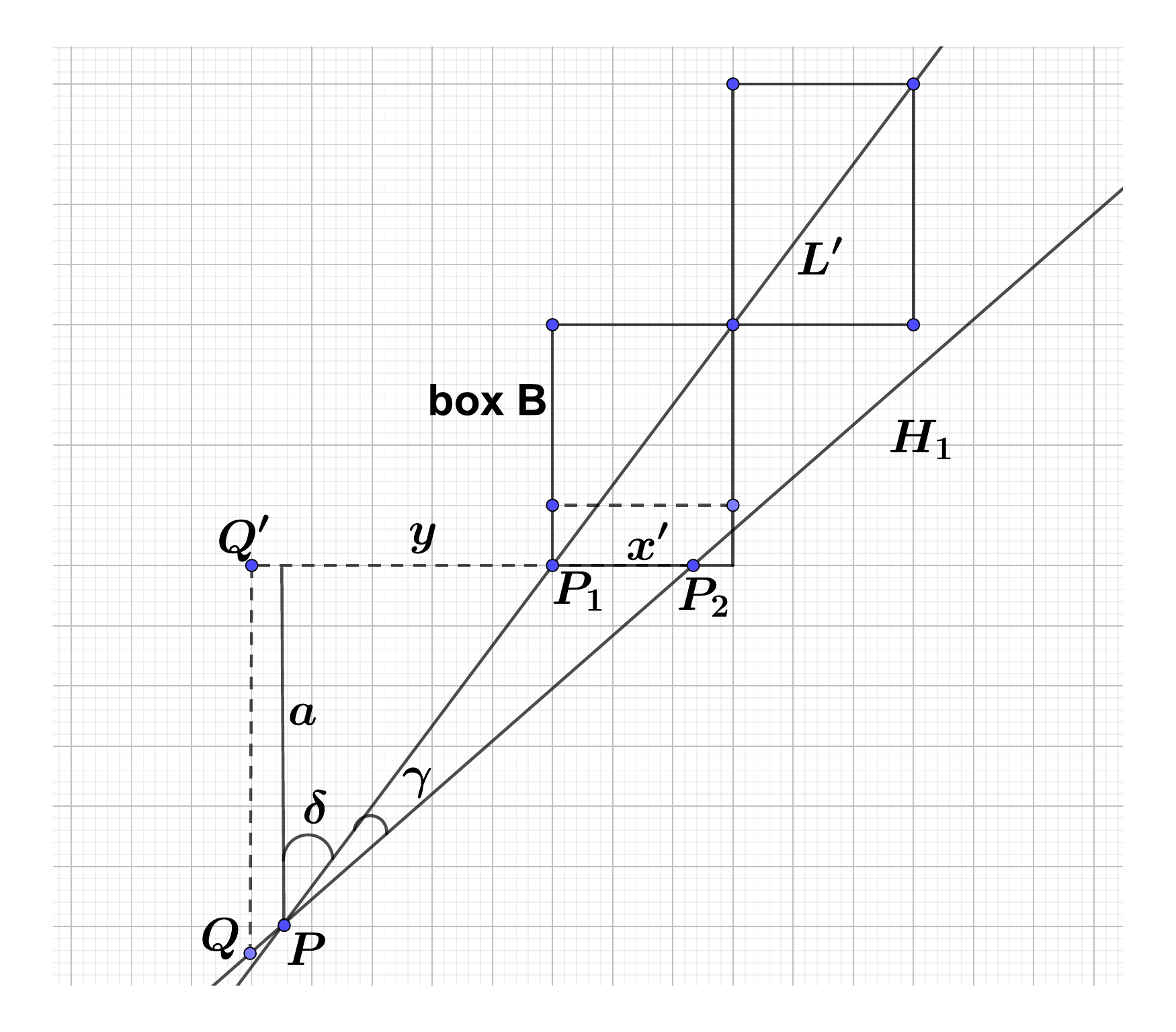}
   \vspace*{-.2in}
    \caption{Exploration of a chain of boxes}
    \label{below}
\end{figure}
Note that since the angle $\gamma +\delta$ is given in the description of the wedge and the slope of line $L'$ has been decided above, we know both angles $\gamma$ and $\delta$ before designing the algorithm.

We have $\tan(\gamma +\delta)=\frac{x'+y}{a}$ and $\tan \delta=\frac{y}{a}$. Let $z=\frac{\tan(\gamma +\delta)}{\tan \delta}-1$.  Hence $y=\frac{x'}{z}$. Let $Y=\lceil x/z \rceil$. 
In view of $x'\leq x$, we have $Y\geq y$. Let $Q'$ be the grid node at distance $Y$ West from $P_1$.
Hence the vertical line containing $Q'$ at distance $Y$ West from $P_1$ intersects the boundary $H_1$ at the point $Q$ South-West of $P$.

The following procedure starts and ends at the South-West corner $P_1$ of box $B$. Upon its completion, all grid nodes in the triangle $QQ'P_2$ are explored. Hence all grid nodes in the triangle $PP_1P_2$ are explored.

\vspace*{0.5cm}

{\bf Procedure} {\tt Below the Chain}

\hspace*{1cm} {\bf for} $i:=1$ {\bf to} $Y$ {\bf do}

\hspace*{2cm} take port $W$;

\hspace*{1cm} $countS :=0$;

\hspace*{1cm} {\bf while} $S$ is free {\bf do}

\hspace*{2cm} take port $S$;

\hspace*{2cm} $countS:=countS+1$;

\hspace*{2cm} $countE :=0$;

\hspace*{2cm} {\bf while} $E$ is free {\bf do}

\hspace*{3cm} take port $E$;

\hspace*{3cm} $countE:=countE+1$;

\hspace*{2cm} {\bf for} $i:=1$ {\bf to} $countE$ {\bf do}

\hspace*{3cm}  take port $W$;

\hspace*{1cm} {\bf for} $i:=1$ {\bf to} $countS$ {\bf do}

\hspace*{2cm} take port $N$;

\hspace*{1cm} {\bf for} $i:=1$ {\bf to} $Y$ {\bf do}

\hspace*{2cm} take port $E$;

\vspace*{0.5cm}

Note that the maximum values of variables $countS$ and $countE$ can be pre-computed knowing angles $\gamma$ and $\delta$, hence an automaton can be designed to perform the above procedure.

Now, starting at the bottom line, the following procedure explores systematically the part of the wedge above the bottom line, by repeatedly executing (forever) the sequence of procedures
{\tt Explore Chain},  {\tt Below the Chain}, {\tt Explore Staircase Up}
 and going one step North. Notice that while going up there is nothing to check: by assumption, the boundaries of the wedge have positive slopes, and hence the port $N$ is always free.

\vspace*{0.5cm}

{\bf Procedure} {\tt Explore Up}

\hspace*{1cm} {\bf repeat forever}

\hspace*{2cm} {\tt Explore Chain};

\hspace*{2cm} {\tt Below the Chain};

\hspace*{2cm} {\tt Explore Staircase Up}; 

\hspace*{2cm} take port $N$

\vspace*{0.5cm}
Now our algorithm can be succinctly formulated as follows.
\vspace*{0.5cm}

{\bf Algorithm} {\tt Explore Obtuse Small Walled}

\hspace*{1cm} {\tt Explore Down} 

\hspace*{1cm} {\tt Explore Bottom Triangles}

\hspace*{1cm} {\tt Explore Up}

\begin{figure}[h]
   \vspace*{-.194in}
     \centering
    \includegraphics[scale = .5]{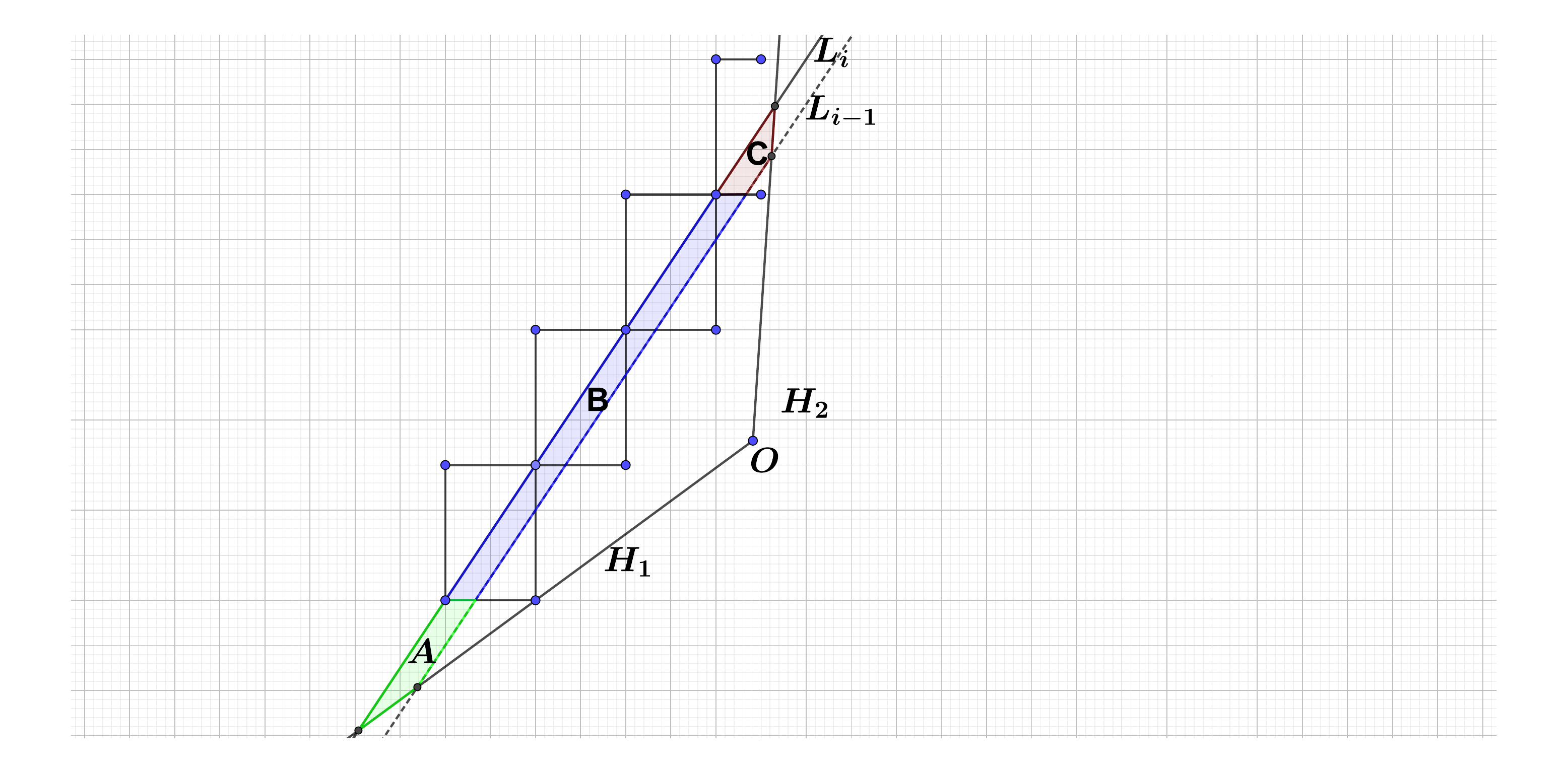}
   \vspace*{-.2in}
    \caption{Part A is explored by procedure Below the Chain, 
part B is explored by procedure Explore Chain, and  part C is explored during the 
execution of procedure Explore Staircase Up }
    \label{th 2.1}
\end{figure}

\begin{theorem}
For every connected small walled wedge there exists an automaton that explores it without pebbles.
\end{theorem}

\begin{proof}
In view of Proposition \ref{acute} we may assume that the wedge is obtuse (hence it is connected) and without loss of generality we may assume that the vector determining $H_1$ has both components negative and the vector determining  $H_2$ has both components positive. Other cases are similar. By the description of procedure {\tt Explore Down}, upon completion of it the agent is at a cardinal node of the bottom line.
It follows from Lemma \ref{triangle} that upon completion of procedure {\tt Explore Bottom Triangles}, all grid nodes in the bottom triangle (bounded by the boundaries of the wedge and by the bottom line) are visited.

Consider the bottom line $L^*$ and a grid node $u$ of it.
Let $L_i$, for $i=0,1,2,\dots$, be the line parallel to $L^*$ and containing the grid node $u_i$ such that $u_i$ is North of $u$ and the distance between them is $i$.
By induction on the number of turns of the loop ``repeat forever'' of procedure {\tt Explore Up}, all nodes of the wedge in the triangle bounded by the boundaries of the wedge and by line $L_i$ are explored after the $i$-th turn of this loop. Indeed, the base case for $i=0$ holds in view of Lemma \ref{triangle}, and the inductive step holds by the descriptions of procedures {\tt Explore Chain},  {\tt Below the Chain} and {\tt Explore Staircase Up}. More precisely, let $R_i$ be the region of the wedge between lines $L_{i-1}$ and $L_i$.
The middle part of $R_i$ is  explored by procedure {\tt Explore Chain}, the bottom part of $R_i$ is explored by procedure {\tt Below the Chain}, and  the top
part of $R_i$ is explored by the last loop {\bf for} $i:=1$ {\bf to} $y$ of procedure  {\tt Explore Staircase Up} (cf. Fig. \ref{th 2.1}).

Thus Algorithm {\tt Explore Obtuse Small Walled} explores the entire wedge. It is straightforward that procedures {\tt Explore Down} and {\tt Explore Up} can be executed by a finite automaton independent of the starting node (although the actual execution of these procedures and the bottom line may depend on the starting node). Since all nodes $A_r$ and sets $\Sigma_r$
in procedure {\tt Explore Bottom Triangles}  can be precomputed without knowing the starting node $v$, this procedure can also be executed by a finite automaton independent of the starting node.
\end{proof}

\subsection{Semi-walled wedges}

In this section we show that the minimum number of pebbles to explore a small semi-walled wedge is 1. We first show how to explore such a wedge using an automaton with one pebble and then we show that an automaton without pebbles cannot accomplish this task. Throughout the section we assume that the boundary half-line $H_1$ is a wall and the boundary half-line $H_2$ is free.
\subsubsection{Exploration with one pebble}

For the positive result, we may assume that the angle between boundaries of the wedge is at least $\pi/2$ because it is enough to show that a wedge wider than required can be explored, and we can modify the free boundary. Hence the wedge is connected. There are two cases: the easier one when boundaries are in adjacent quadrants of the plane and the more difficult case when they are in opposite quadrants or when one of the boundaries is either vertical or horizontal.  In the easier case there exists a cutting line of the wedge that is either vertical or horizontal which facilitates the moves of the agent.

\vspace*{0.5cm}

{\bf Boundaries in adjacent quadrants}

\vspace*{0.5cm}

Without loss of generality we assume that both boundaries are North of the origin $O$, $H_1$ has a negative and $H_2$ a positive slope. Hence the wedge has a horizontal cutting line. The algorithm for the other cases is similar. Moreover, by possibly widening the wedge, we can choose the slope of the free boundary $H_2$ to be $1/x$, for a positive integer $x$. If the origin $O$ is a grid node then $H_2$ is a rational $(x,1)$-line. In this case we define line $P$ to be $H_2$ (see Fig \ref{fig-5} (A) ). Otherwise, line $P$ is defined as the line with slope $1/x$ containing the South-East vertex of the unit square of the grid inside which the origin $O$ is located (see Fig \ref{fig-5} (B)-(C)). By definition, $P$ is an $(x,1)$-line and every node of it on a horizontal cutting line of the wedge is located either in the free boundary $H_2$ or East of it. Define the node $O'$ of the grid as follows. It is the origin $O$ of the wall, if $O$ is a node of the grid, and otherwise it is the North-East vertex of the unit square of the grid inside which the origin $O$ is located; we consider the West and the South sides of a unit square as belonging to it  (see Fig \ref{fig-5}). Define the node $O''$ as the western neighbor of~$O'$.

\begin{figure}[h]
   \vspace*{-.194in}
     \centering
    \includegraphics[scale = .60]{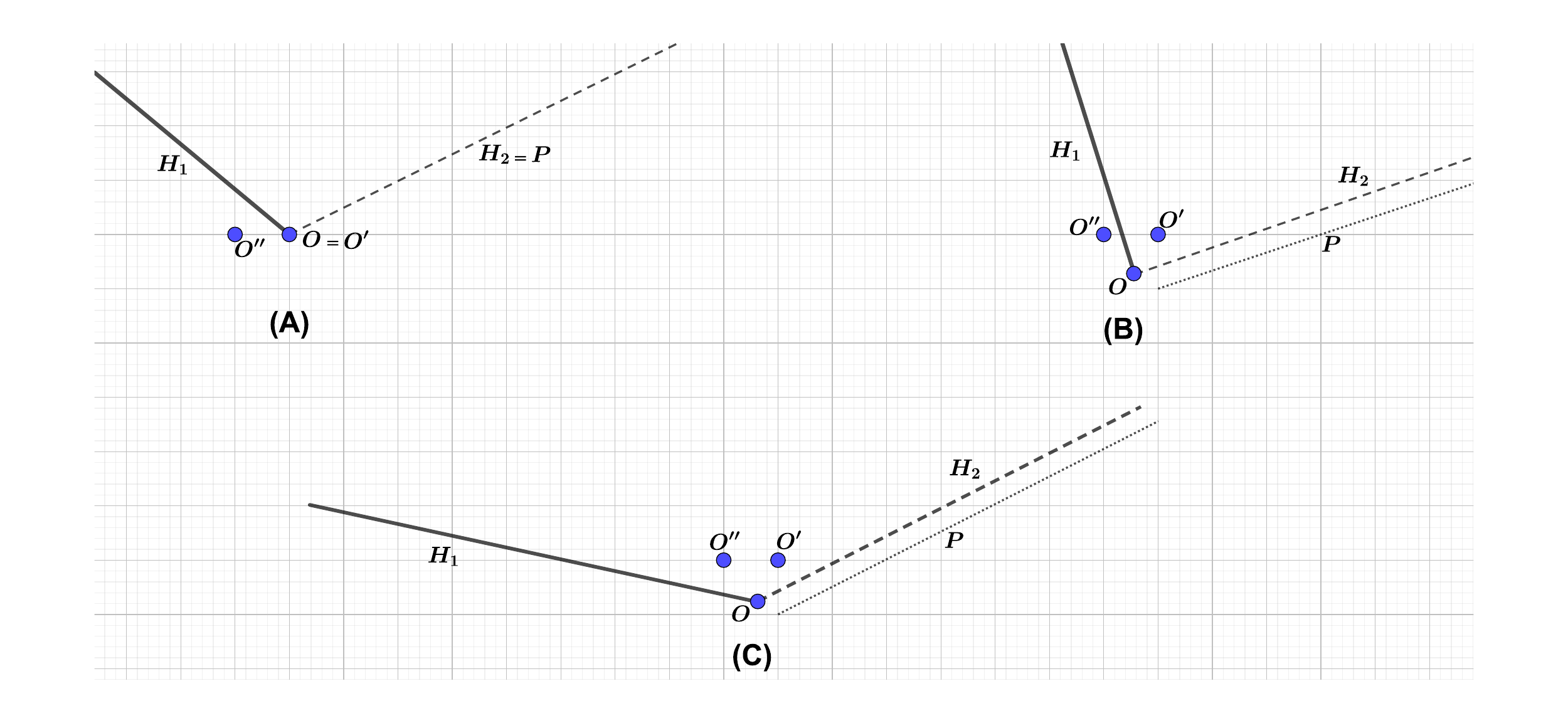}
   \vspace*{-.2in}
    \caption{(A) $O$ is a node of the grid, hence $O=O'$ and $O''$ is not in the wedge; procedure {\tt Go to Origin ends at node $O'$}. (B) $O$ is not a node of the grid, $O''$ is not in the wedge; procedure {\tt Go to Origin ends at node $O'$}. (C) $O$ is not a node of the grid, $O''$ is in the wedge; procedure {\tt Go to Origin ends at node $O''$}.  }
    \label{fig-5}
\end{figure}

The high-level idea of the algorithm is the following. 
The aim of the first part (formalized as procedure {\tt Go to Origin}) is to go West from the starting node until ``hitting'' the wall $H_1$ \footnote{``hitting the wall'' means that the agent tries a port which is blocked by the wall.}
 and then to ``slide down'' along $H_1$ (actually repeatedly going East and South) until reaching the grid node $O'$ or $O''$. In the second part, the pebble is used. This part  works in phases. In each phase, the segment of the current horizontal line contained in the wedge is explored as follows. A phase starts with the agent at a node which is on a horizontal cutting line $L$ of the wedge and on or East of the line $P$. The agent drops the pebble,  goes West until the wall and goes back to the pebble. Then it picks the pebble, goes one step North and $x$ steps East which brings it to a node on the next horizontal cutting line and on or East of the line $P$. Then it drops the pebble and a new phase (identical to the previous one) starts.

The first part of the algorithm is formalized as the following procedure. 

\vspace*{0.5cm}

{\bf Procedure} {\tt Go to Origin}

\hspace*{1cm} {\bf while} $W$ is free {\bf do}

\hspace*{2cm} take port $W$;

\hspace*{1cm} {\bf while} $W$ is blocked {\bf do}

\hspace*{2cm} {\bf while} $S$ is blocked {\bf do}

\hspace*{3cm} take port $E$;

\hspace*{2cm} take port $S$;

\hspace*{1cm} take port $N$;

\hspace*{1cm} {\bf if} $W$ is free {\bf then}

\hspace*{2cm} take port $W$;

\vspace*{0.5cm}

The ``while $W$ is free'' loop makes the agent hit the wall, and the  ``while $W$ is blocked'' loop makes it ``slide down'' along the wall. 
After leaving this loop the agent goes either to node $O'$ or to $O''$. More precisely it goes to $O'$, if $O''$ is outside of the wedge and it goes to $O''$ otherwise. 

The next procedure starts where the previous one ended, i.e.,  at node $O'$ or $O''$, and it implements the second part of the algorithm.

\vspace*{0.5cm}

{\bf Procedure} {\tt Bounce to Pebble}

\hspace*{1cm} {\bf for} $i:=1$ {\bf to} $x+1$ {\bf do}

\hspace*{2cm} take port $E$;

\hspace*{1cm} drop a pebble;

\hspace*{1cm} {\bf repeat forever}

\hspace*{2cm} {\bf while} port $W$ is free {\bf do}

\hspace*{3cm} take port $W$;

\hspace*{2cm} {\bf while} no pebble on current node {\bf do} 

\hspace*{3cm} take port $E$;

\hspace*{2cm} pick a pebble;

\hspace*{2cm} take port $N$;

\hspace*{2cm} {\bf for} $i:=1$ {\bf to} $x$ {\bf do}

\hspace*{3cm} take port $E$;

\hspace*{2cm} drop a pebble;

\vspace*{0.5cm}

The first ``for'' loop places the agent and the pebble on the horizontal line containing nodes $O'$ and $O''$ and either on the line $P$ or East of it. 
More precisely, the first case happens when the agent starts the procedure at node $O''$ and the second case happens when it starts at $O'$.
The ``repeat forever'' loop executes the consecutive phases of the second part of the algorithm.

Now the algorithm is a combination of the two above procedures.

\vspace*{0.5cm}

{\bf Algorithm} {\tt Explore Small Semi-Walled -- Adjacent Quadrants}

\hspace*{1cm} {\tt Go to Origin}

\hspace*{1cm}  {\tt Bounce to Pebble}

\begin{lemma}\label{adjacent}
For any small semi-walled wedge with boundaries in adjacent quadrants there exists an automaton that explores it with one pebble.
\end{lemma}

\begin{proof}
By the description of procedure {\tt Go to Origin}, upon completion of this procedure, the agent is at node $O'$ or at node $O''$. After the first `for'' loop of procedure  {\tt Bounce to Pebble}  the agent and the pebble are on the southern-most horizontal line of the wedge, on or East of its free boundary $H_2$. After each turn of the ``repeat forever'' loop, the following invariant is preserved: the agent and the pebble are on the next 
 horizontal line (one step North of the preceding one) of the wedge, on or East of its free boundary $H_2$. In each turn of this loop the agent explores the segment of the current horizontal line between the wall and the pebble. Hence, in each turn of the loop it explores the entire segment of the wedge contained in the current horizontal line. It follows that the entire wedge is eventually explored. Since the integer $x$ is known in advance, the algorithm can be executed by a finite automaton with one pebble, independent of the starting node.
\end{proof}

{\bf Boundaries in opposite quadrants}

\vspace*{0.5cm}

Now we consider the more difficult case when the boundaries are in opposite quadrants (we include the limit case when one of the boundaries is either vertical or horizontal).
Consider a wedge with boundaries $H_1$ and $H_2$. Without loss of generality assume that the vector determining $H_1$ has both components negative and the vector determining $H_2$ has both components positive. The algorithm for other cases is similar.

%Suppose that the line $L$ passes through nodes $(a,b)$ and $(a+x,b+y)$. Hence $L$ is a $(x,y)$-line (see Fig. \ref{fig-4}). 
%For any $(x,y)$-line $L'$ parallel to $L$ and any node $w=(a',b')$ of the grid contained in $L'$, we define the {\em box} of $w$, as the set of all nodes $(a'-x',b'-y')$ of the grid, such that $0\leq x'\leq x$ and $0 \leq y' \leq y$. Exploring a box means visiting all nodes of it.
%The {\em chain of boxes} for $L'$ and $w$ is the set of boxes for nodes $(\dots, a'-2x,b'-2y),(a'-x,b'-y), (a',b'), (a'+x,b'+y), (a'+2x,b'+2y),\dots)$.

The high-level idea of the algorithm is the following. 
The agent starts at some node $v$ of the wedge. 
Let $L$ be a rational cutting line of the wedge passing through node $v$. Suppose that $L$ is an $(x,y)$-line.
The first part of the algorithm is similar to procedure {\tt Go to Origin} and its aim is to get to a node of the grid close to $O$. Now it is impossible to achieve this by going straight in order to hit the wall, as we do not know where in the wedge the agent starts and neither a horizontal nor a vertical trip is guaranteed to hit the wall. Instead, the agent ``slides down'' to the wall following a staircase of line $L$. After hitting the wall it ``slides up'' along the wall (since the slope of the wall is positive), to get close to $O$. The first part of the algorithm is concluded as follows. The agent starts at a node $O^*$ of the grid close to $O$. Let $L_0$ be the line parallel to $L$ passing through $O^*$. The agent visits all grid nodes of the triangle bounded by the boundaries and by $L_0$ and gets back to $O^*$.

The second part of the algorithm is executed in phases and the pebble is used.
Let $L_i$, for $i=0,1,2,\dots,$ be lines parallel to $L$, such that $L_{i+1}$ is at distance 1 West from $L_i$. Line $L_0$ is as above.
Before phase 1, all grid nodes of the triangle bounded by the boundaries and the line $L_0$ are explored, and the pebble is dropped on $L_0$, East of or on the boundary $H_2$.
We want to keep the invariant that after phase $i\geq 1$, all grid nodes of the triangle bounded by the boundaries and the line $L_i$ are explored and the pebble is on line $L_i$, East of or on the boundary $H_2$.
Phase $i\geq 1$ starts with the agent at a node $(c,d)$ which is on line $L_{i-1}$, East of or on the boundary $H_2$. The pebble is at this node.
The agent picks the pebble and goes to a precomputed node $(c',d')$ which is on line $L_{i}$,  East of or on the boundary $H_2$. 
%On its way to $(c',d')$, the agent explores boxes of line $L_{i-1}$ up to the horizontal level of $(c',d')$.
Then the agent drops the pebble at $(c',d')$ and  explores the part of the chain of boxes between the pebble and the wall using procedure {\tt Explore Chain} from the previous section.
Next it executes procedure {\tt Below the Chain} from the previous section to explore the lower part of the triangle between line $L_i$ and the wall $H_1$.
 Then it goes back to the pebble using the staircase which is the upper limit of the chain, and phase $i+1$ starts.

In order to implement the above idea we define a number of procedures. The first procedure brings the agent  close to the wall by ``sliding down'' the staircase of the cutting line $L$. Upon its completion, the agent is at a node with port $S$ blocked by the wall.

\vspace*{0.5cm}

{\bf Procedure} {\tt Slide Down to Wall}

$count :=y$;

\hspace*{1cm}{\bf while}  $count=y$ {\bf do}

\hspace*{2cm} {\bf for} $i:=1$ {\bf to} $x$ {\bf do}

\hspace*{3cm} take port $W$;

\hspace*{2cm} $count :=0$;

\hspace*{2cm} {\bf while} port $S$ free and $count <y$ {\bf do}

\hspace*{3cm} take port $S$;

\hspace*{3cm} $count := count +1$;

\vspace*{0.5cm}

The next procedure starts where the previous one ended and it brings the agent either to node $O'$ or $O''$, where these nodes were defined before. More precisely, the agent ends up at node $O'$, if $O$ is a grid node (in which case $O=O'$), and it ends up at $O''$ otherwise.

\vspace*{0.5cm}

{\bf Procedure} {\tt Slide Up to Origin}

\hspace*{1cm} {\bf while} $S$ is blocked {\bf do}

\hspace*{2cm} {\bf while} $E$ is blocked {\bf do}

\hspace*{3cm} take port $N$;

\hspace*{2cm} take port $E$;

\hspace*{1cm} take port $W$;

\vspace*{0.5cm}

\vspace*{0.5cm}

\begin{figure}[h]
   \vspace*{-.194in}
     \centering
    \includegraphics[scale = .6]{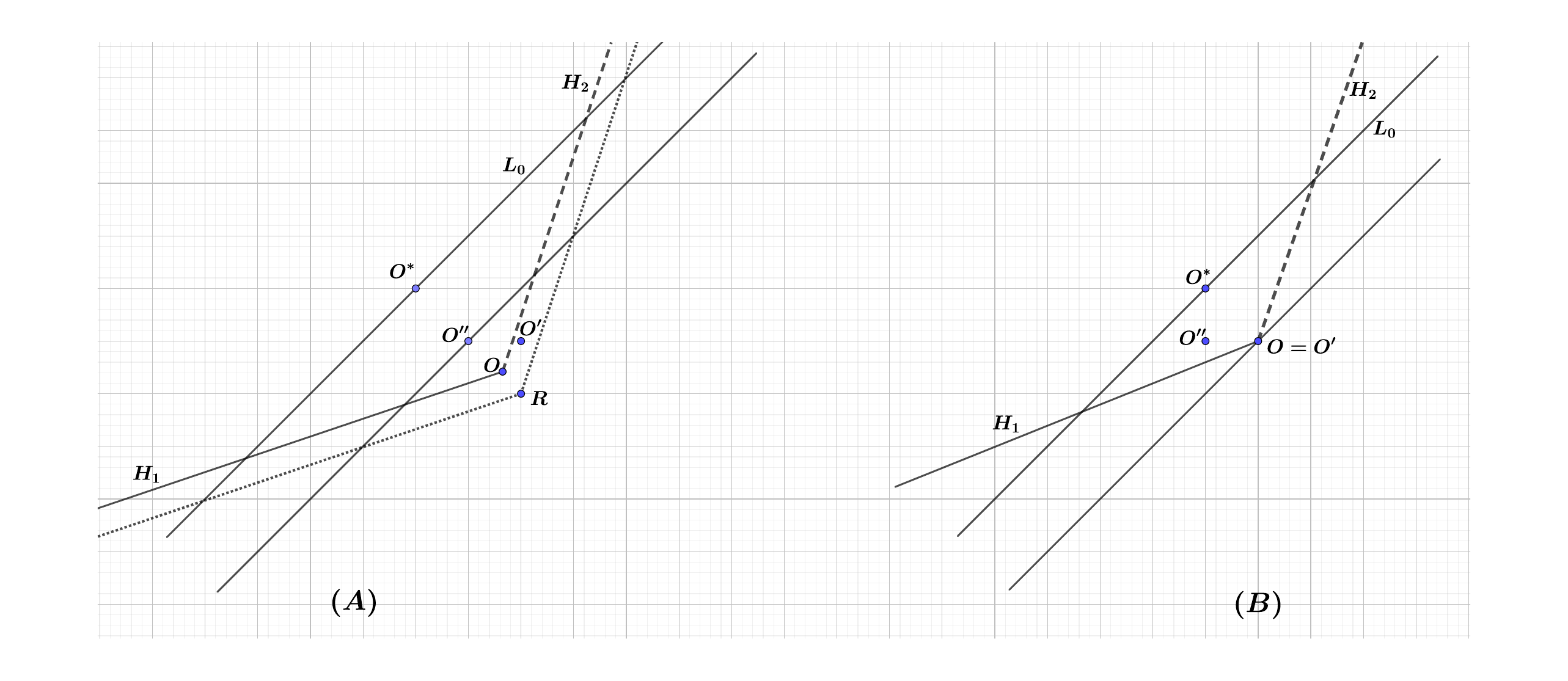}
   \vspace*{-.2in}
    \caption{ (A) $O$ is not a grid node  (B) $O$ is a grid node}
    \label{bottom}
\end{figure}

The aim of the next procedure is to visit  all grid nodes of the triangle bounded by the boundaries of the wedge and by the line $L_0$ parallel to $L$. First we need to precisely define this line. 
After completion of procedure {\tt Slide Up to Origin}, the agent is either at node $O'$, if $O$ is a grid node (in which case $O=O'$), or at $O''$ (see Fig \ref{bottom}). The agent does not know what is the actual situation and has to prepare for both of them. We define the node $O^*$ as follows. If $O$ is a grid node then $O^*$ is the grid node at grid distance 2 from $O$ going North and then West from $O$.
If $O$ is not a grid node then $O^*$ is the grid node at grid distance 2 from $O''$ going North and then West from $O''$. In both cases $L_0$ is defined as the line parallel to $L$ passing through $O^*$. Notice that $L_0$ is always a cutting line of the wedge.

Define $T$ as the triangle between the boundaries of the wedge and $L_0$. Our aim now is to visit all grid nodes of $T$. Unfortunately, the agent does not know exactly what is $T$, as this depends on where $O$ is situated. We want to cover all possible cases, similarly as in procedure {\tt Explore Bottom Triangles}. If $O$ is a grid node then the triangle $T$ is exactly determined. In this case define $T'=T$. Otherwise, the agent cannot determine $T$  exactly but $T$ is certainly included in the following triangle $T'$. Let $R$ be the South-West corner of the unit grid square containing $O$.  $T'$ is bounded by $L_0$ and by the lines parallel to the wedge boundaries and passing through $R$ (see Fig. \ref{bottom}).

In order to visit all grid nodes in $T$, it is enough to visit all grid nodes in $T'$. Let $P$ be the node in which the agent concluded procedure {\tt Slide Up to Origin}. Then the agent takes port $N$ and then port $W$ which brings it to node $O^*$ in all cases. Now in order to visit all nodes of $T'$ in all cases, we proceed similarly as in procedure {\tt Explore Bottom Triangles}. Note that, when the agent is at $O^*$, the triangle $T'$ is exactly defined in both cases. Hence the set of grid nodes in $T'$ can be determined for each case (knowing the boundaries slopes and the slope of $L$) but the agent does not know which case is the real one. Let $\Sigma_1$ be the set of grid nodes  in $T'$ in the case when $O$ is a grid node, and let $\Sigma_2$ be the set of grid nodes  in $T'$ in the case when $O$ is not a grid node.
The agent can determine both  $\Sigma_1$ and $\Sigma_2$. Hence the following procedure, starting at node $P$ and using procedure {\tt Visit} introduced in the previous section,  always visits all grid nodes of the triangle $T'$ and finishes at $O^*$.

\vspace*{0.5cm}

{\bf Procedure} {\tt Bottom Triangle}

\hspace*{1cm} take port $N$;

\hspace*{1cm} take port $W$; 

\hspace*{1cm} let $O^*$ be the current node;

\hspace*{1cm} {\tt Visit} $(O^*, \Sigma_1)$;

\hspace*{1cm} {\tt Visit} $(O^*, \Sigma_2)$;

\vspace*{0.5cm}

The aim of the next procedure, called after procedure {\tt Bottom Triangle}, is to drop the pebble at its first location. This location should be on the line $L_0$, East of or on the boundary $H_2$, so as to satisfy our invariant at the very beginning of the second part of the algorithm. Let $I$ be the segment of the line $L_0$ that is the side of the triangle $T'$. It is enough to drop the pebble at any node of the line $L_0$ above the upper end of the segment $I$. Although the actual length of $I$ is unknown to the agent, it is easy to precompute an upper bound on it. Both possible triangles $T'$ are similar and the one in the case when $O$ is not a node of the grid is larger. Hence the length $z$ of its side parallel to $L$ (which is easy to compute knowing the slopes of all sides of $T'$) can be used as the desired upper bound.
It follows that if the agent goes up on line $L_0$ starting at node $O^*$ at distance at least $z$, it will be East of or on the boundary $H_2$. Of course, the agent cannot really go on $L_0$ but it can follow a staircase of $L_0$ for a number of steps easy to precompute. Let $r=\lceil z/y \rceil$. Going $r$ steps of the staircase of $L_0$ up from $O^*$ is enough. Hence the following procedure places the pebble on line $L_0$, East of or on the boundary $H_2$.

\vspace*{0.5cm}

{\bf Procedure} {\tt Initialize Pebble}

\hspace*{1cm} {\bf for} $j:=1$ {\bf to} $r$ {\bf do}  

\hspace*{2cm} {\bf for} $i:=1$ {\bf to} $y$ {\bf do} 

\hspace*{3cm} take port $N$;

\hspace*{2cm} {\bf for} $k:=1$ {\bf to} $x$ {\bf do}  

\hspace*{3cm} take port $E$;

\hspace*{1cm} drop a pebble;

\vspace*{0.5cm}

The next procedure starts at a node $T_1$ of some line $L_{i-1}$, East of or on the boundary $H_2$, with the pebble at this node. The aim of the procedure is to carry the pebble to a node $T_2$ of line $L_i$,
East of or on the boundary $H_2$. 
%On the way to the new location $T_2$, the agent explores the part of the chain of boxes of $L_i$ between the node $T_1$ and the horizontal level of $T_2$.

Let $Q$ be the intersection point of line $H_2$ and $L_i$. Let $P_2$ be a grid node on $L_i$ such that $P_2$ is at distance at least 1 East from the boundary $H_2$. Let $P_1$ be the intersection point of $H_2$ with the horizontal line passing through $P_2$. Hence the distance $t$ between $P_1$ and $P_2$ is at least 1. Let $Q'$ be the point on this horizontal line and on the vertical line passing through $Q$. Let $z$ be the distance between $Q'$ and $P_1$, and let $Y$ be the distance between $Q$ and $Q'$. Let $\gamma$ be the angle between the vertical line and line $H_2$, and let $\delta$ be the angle between $H_2$ and $L$ (see Fig. \ref{move pebble}). Note that we know these angles and can use them in our procedure.

 \begin{figure}[h]
   \vspace*{-.194in}
     \centering
    \includegraphics[scale = .5]{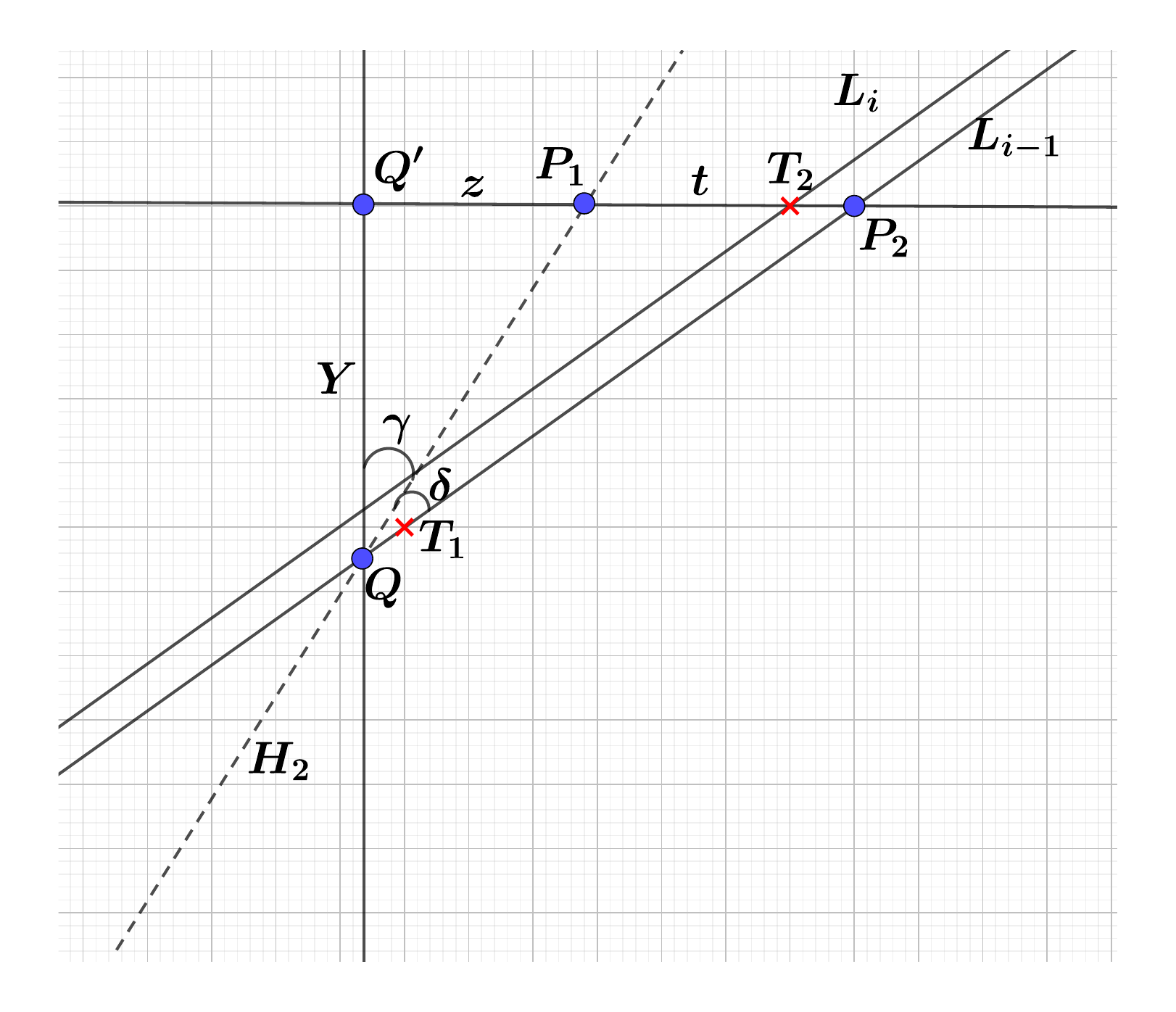}
   \vspace*{-.2in}
    \caption{Moving the pebble to the next location}
    \label{move pebble}
\end{figure}

Our aim is to put the pebble  at node $T_2$ at distance 1 West from $P_2$. The node $T_2$ is on line $L_{i}$. We need to find $Y$ sufficiently large to guarantee $t\geq 1$ because then the pebble will be East of or on the boundary $H_2$, as desired. We have $\tan\gamma=\frac{z}{Y}$ and $\tan(\gamma +\delta)=\frac{z+t}{Y}$. Hence $\frac{\tan(\gamma +\delta)}{\tan\gamma}=\frac{z}{t}+1$. Define
$\sigma=\frac{\tan(\gamma +\delta)}{\tan\gamma}-1$. Hence $t=\frac{z}{\sigma}=\frac{Y\tan \gamma}{\sigma}$. In order to guarantee $t\geq 1$ it is enough to take $Y\geq \frac{\sigma}{\tan \gamma}$.
This will hold if we carry the pebble $m=\lceil \sigma/(y \tan \gamma) \rceil$ steps of the staircase up on line $L_{i-1}$. Notice that the integer $m$ depends on angles $\gamma$ and $\delta$ and on the integer $y$
which is the height of a step of the staircase, and hence $m$ can be precomputed.

\vspace*{0.5cm}

{\bf Procedure} {\tt Move Pebble}

\hspace*{1cm} pick a pebble

\hspace*{1cm} {\bf for} $i:=1$ {\bf to} $m$ {\bf do} 

\hspace*{2cm} {\bf for} $j:=1$ {\bf to} $y$ {\bf do}  

\hspace*{3cm} take port $N$;

\hspace*{2cm} {\bf for} $k:=1$ {\bf to} $x$ {\bf do}  

\hspace*{3cm} take port $E$;

\hspace*{1cm} take port $W$;

%\hspace*{2cm} {\tt Explore Box};
%
%\hspace*{2cm} {\bf for} $j:=1$ {\bf to} $y$ {\bf do}  
%
%\hspace*{3cm} take port $N$;
%
%\hspace*{2cm} {\bf for} $k:=1$ {\bf to} $x$ {\bf do}  
%
%\hspace*{3cm} take port $E$;
%

\hspace*{1cm} drop a pebble; 

\vspace*{0.5cm}

The next procedure will be called upon completion of procedures {\tt Explore Chain} and {\tt Below the Chain} from the previous section. Its aim is to get back to the pebble by going up the staircase of the current cutting line.

\vspace*{0.5cm}

{\bf Procedure} {\tt Go Back to Pebble}

\hspace*{1cm} {\bf while} no pebble on current node {\bf do}

\hspace*{2cm} {\bf for} $i:=1$ {\bf to} $y$ {\bf do} 

\hspace*{3cm} take port $N$;

\hspace*{2cm} {\bf for} $j:=1$ {\bf to} $x$ {\bf do}  

\hspace*{3cm} take port $E$;

\hspace*{1cm} {\bf for} $i:=1$ {\bf to} $y$ {\bf do} 

\hspace*{2cm} take port $N$;

\hspace*{1cm} {\bf for} $j:=1$ {\bf to} $x$ {\bf do}  

\hspace*{2cm} take port $E$;

\vspace*{0.5cm}

Now the algorithm can be succinctly formulated as follows.

\vspace*{0.5cm}

{\bf Algorithm} {\tt Explore Small Semi-Walled -- Opposite Quadrants}

\hspace*{1cm} {\tt Slide Down to Wall}

\hspace*{1cm} {\tt Slide Up to Origin}

\hspace*{1cm} {\tt Bottom Triangle}

\hspace*{1cm} {\tt Initialize Pebble}

\hspace*{1cm} {\bf repeat forever}

\hspace*{2cm} {\tt Move Pebble} 

\hspace*{2cm} {\tt Explore Chain}  

\hspace*{2cm} {\tt Below the Chain}

\hspace*{2cm} {\tt Go Back to Pebble}

\begin{lemma}\label{opposite}
For any small semi-walled wedge with boundaries in opposite quadrants there exists an automaton that explores it with one pebble.
\end{lemma}

\begin{proof}
Let $T_i$, for $i=0,1,\dots$,  be the triangle bounded by line $L_i$ and by the boundaries of the wedge.
After the execution of procedures {\tt Slide Down to Wall} and {\tt Slide Up to Origin}, the agent is either at node $O$ or $O''$. After the execution of procedure {\tt Bottom Triangle}, all the grid nodes in the triangle $T_0$ are visited. After the execution of procedure  {\tt Initialize Pebble}, the pebble is placed on line $L_0$, East of or on the boundary $H_2$. 
The algorithm works in phases $1,2,\dots$, determined by the ``repeat forever'' loop.
Let $Inv(i)$ be the following invariant: 
\begin{quotation}
At the beginning of phase $i$, the pebble is on line $L_{i-1}$, East of or on the boundary~$H_2$, and all grid nodes in the triangle $T_{i-1}$ are visited.
\end{quotation}
Hence $Inv(1)$ holds. Assume by induction that $Inv(i)$ holds for some $i\geq 1$.
In phase $i$, after the execution of procedures  {\tt Move Pebble}, {\tt Explore Chain}  and {\tt Below the Chain},  the grid nodes of the wedge in the stripe between lines $L_{i-1}$ and $L_i$ are visited. Hence at this point, all grid nodes in  $T_{i}$ are visited.
After the execution of procedure {\tt Go Back to Pebble}, the pebble is placed on line $L_{i}$, East of or on the boundary~$H_2$ (cf. Fig. \ref{move pebble}). Hence $Inv(i+1)$ is satisfied. By induction, $Inv(i)$ is satisfied, for all $i\geq 1$. Since the wedge is the union of sets of grid nodes in all triangles $T_i$, this proves that Algorithm {\tt Explore Small Semi-Walled -- Opposite Quadrants} explores any small
semi-walled wedge with boundaries in opposite quadrants. In order to conclude the proof, it remains to show that this algorithm can be executed by an automaton with a pebble. This follows from the observation that all procedures used by the algorithm depend only on integer parameters that can be precomputed knowing the slopes of the boundaries of the wedge and knowing the slope of the cutting line $L$ that, in turn, has been chosen using the slopes of the boundaries. Hence, given the wedge as a couple of vectors of its boundaries, the appropriate automaton can be constructed, independent of the starting node. Since it uses only one pebble, this proves the lemma.
\end{proof}

Lemmas \ref{adjacent} and \ref{opposite} imply the main positive result of this section.

\begin{theorem}\label{th small semi}
For any small semi-walled wedge there exists an automaton that explores it with one pebble.
\end{theorem}

\subsubsection{Impossibility of exploration without a pebble}

In this section we show that the number of pebbles in Theorem \ref{th small semi} cannot be decreased.
We will use the following general lemma.

\begin{lemma}\label{infinitely}
Consider the trajectory $T$ of an automaton that explores a wedge $W$ without a pebble.
Then no node can be visited infinitely many times.
\end{lemma}

\begin{proof}
We prove the lemma by contradiction. Suppose that some node is visited infinitely many times.
Hence it is visited at least twice in the same state $S$. Let $t_1$ and $t_2$ be the first two steps when this occurs. Let $\tau$ be the part of the trajectory $T$ between steps $t_1$ and $t_2$.
 Let $x$ be the length of $\tau$. Hence $t_2=t_1+x$. The trajectory $\tau$ will be repeated forever starting at the first visit in step $t_1$. Indeed, by induction on $i$, the agent will be at the same node and in the same state in any step $t_1+kx+i$, for any $i<x$ and any natural number $k$. It follows that the number of nodes visited by the agent is at most $t_1+x$, and hence the agent cannot explore the entire wedge. This contradiction proves the lemma.
\end{proof}

\begin{theorem}\label{th-lb-semi-walled}
No semi-walled wedge can be explored by an automaton without any pebble.
\end{theorem}

\begin{proof}
Consider any semi-walled wedge $W$ and any automaton with $s$ states exploring it without any pebble.
Let $T$ be the (infinite) trajectory of this automaton starting at any node of $W$.
Without loss of generality assume that the wall $H$ of $W$ starts at the node $O$ of the grid and that the wall goes vertically North. The proof in other cases is similar.
Let $H'$ be the vertical half-line that goes South of $O$.
We will use the following claim.

\vspace*{0.5cm}

 \noindent
{\bf Claim.}
After visiting the node $O$ the agent can never reach the horizontal line at distance $s+1$ South of $O$.

We prove the claim by contradiction. 
Let $P_i$, for $i=1,2,\dots, s+1$, be the horizontal line at distance $i$ South of $O$.
Suppose that the agent visits node $O$ and then reaches the line $P_{s+1}$. 
Then there exist grid nodes $w_i$ in $P_i$ visited in step $t_i$, such that $t_j>t_i$, for $j>i$, and the agent never visits a node of $P_1$ after step $t_1$.
It follows that there are grid nodes $v'=w_i$  and $v''=w_j$ for some $j>i$, such that the agent is in the same state $S$ when visiting node $v'$ in step $t'=t_i$ and node $v''$ in step $t''=t_j$.

 Let $v_0,v_1,\dots, v_r$, with $v_0=v'$ and $v_r=v''$,  be the sequence of nodes visited between steps $t'$ and $t''$. Hence $t''=t'+r$.  Let $\alpha$ be the vector $(v',v'')$. Hence in step $t'+kr+j$, for any natural $k$ and any $j<r$, the automaton will be in the same state $S_j$ at node $v_j+k\alpha$. (Intuitively, the trajectory $(v_0,v_1,\dots, v_r)$ is shifted infinitely many times, from step $t'$ on, by multiples of the vector $\alpha$). Let $L_1$ and $L_2$ be two lines parallel to vector $\alpha$, such that all nodes $v_0,v_1,\dots, v_r$ are between these lines. It follows that all nodes of the part of the trajectory $T$ after step $t'$ are between lines $L_1$ and $L_2$ (see Fig. \ref{Proof}). Consequently, at most $t'$ nodes of the trajectory $T$ are not between these lines. Since infinitely many nodes of $W$ are not between lines $L_1$ and $L_2$, there exist nodes of $W$ which will never be visited. This contradiction proves the claim.
 
 \begin{figure}[h]
   \vspace*{-.194in}
     \centering
    \includegraphics[scale = .6]{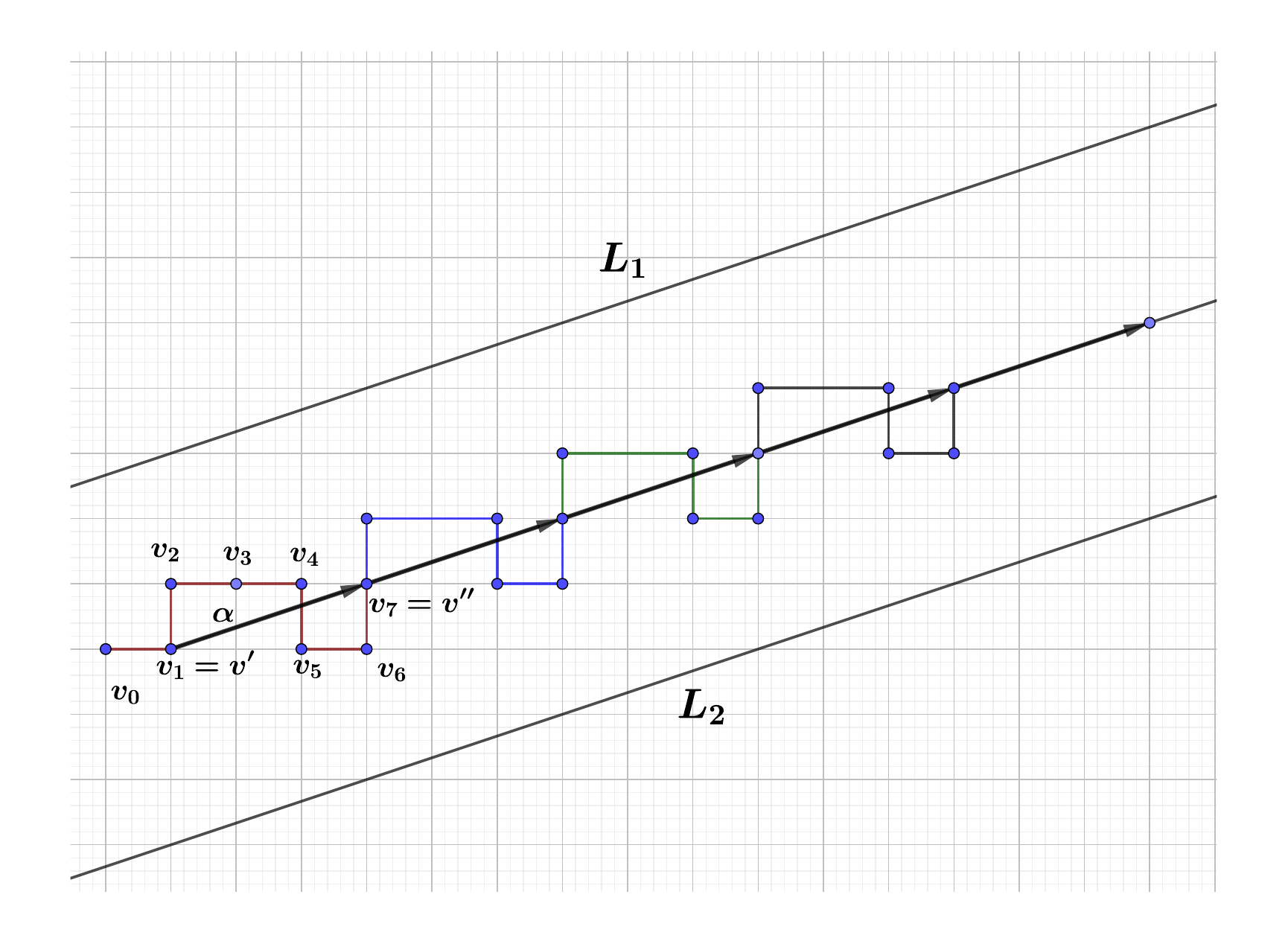}
   \vspace*{-0.1in}
    \caption{Illustration for the proof of Theorem \ref{th-lb-semi-walled} }
    \label{Proof}
\end{figure}

 By Lemma \ref{infinitely},  the node $O$ can be visited only finitely many times. Consider the part $T^*$ of the trajectory $T$ after the last visit of $O$. Note that nodes of the half-line $H'$ can be only visited finitely many times. Indeed, the Claim implies that only nodes of $H'$ at distances at most $s$ from $O$ can be visited after the last visit of $O$. So if there were infinitely many such visits, then some node of $H'$ would be visited infinitely many times, contradicting Lemma \ref{infinitely}. Thus we may consider the part $T'$ of the trajectory $T$ after the last visit of $O$ and of any node of $H'$. 
 
 Since all nodes of the wall $H$ must be visited and there are infinitely many of them, there are two visits of nodes $w'$ and $w''$ of the wall in the same state $S$, in steps $d'$ and $d''$, respectively, such that $d''>d'$ and $w''$ is North of $w'$.
 Consider the part $T''$ of the trajectory $T'$ between steps $d'$ and $d''$. By a similar argument as above, trajectory $T''$ is shifted infinitely many times, from step $d'$ on, by multiples of the vector $(w',w'')$. Note that in order to prove it, we crucially use the fact that  during trajectory $T'$ neither $O$ nor any node of $H'$ is visited anymore. Indeed, if this was possible, it could happen that some node of trajectory $T'$  is a node of $H'$ which has all ports free but the corresponding node of a shift North of this trajectory is a node of the wall and hence, although by induction hypothesis they are visited in the same state, the input is different and the rest of the shifted trajectory could be potentially different. This cannot happen, as among nodes of the vertical line containing $O$ only nodes of the wall different from $O$ can be visited during trajectory $T'$. Thus corresponding nodes of any shift of $T''$ by a multiple of the vector $(w',w'')$ are visited in the same state and give the agent the same input. Thus induction works.
 
 Now the proof is concluded similarly as before. Let $L$ be the vertical line containing the wall and let $L'$ be the vertical line such that all nodes of the trajectory $T''$ are between $L$ and $L'$. Since the vector $(w',w'')$ is oriented North, all nodes of the trajectory $T$ after step $d'$ are  between $L$ and $L'$. Consequently, at most $d'$ nodes of the trajectory $T$ are not between these lines. Since infinitely many nodes of $W$ are not between lines $L$ and $L'$, there exist nodes of $W$ which will never be visited. This contradiction proves the theorem.
\end{proof}

\begin{corollary}
The minimum number of pebbles sufficient to explore any small semi-walled wedge is one.
\end{corollary}

\subsection{Free wedges}

In this section we show that the minimum number of pebbles to explore a small free wedge is 2. We first show how to explore such a wedge using an automaton with two pebbles and then we show that an automaton with only one pebble cannot accomplish this task. 

\subsubsection{Exploration with two pebbles}

Consider any small free wedge with boundaries $H_1$ and $H_2$. Recall that in the case of free wedges we assume that the origin $O$ of the wedge is at the starting node of the agent. 
For the positive result, we may assume that the angle between boundaries of the wedge is at least $\pi/2$ because it is enough to show that a wedge wider than required can be explored, and we can modify  free boundaries. Hence the wedge is connected. 
As before, there are two cases: the easier case when the boundaries are in adjacent quadrants,
and the more difficult case when they are in opposite quadrants or when one of the boundaries is either vertical or horizontal.  In the easier case there exists a cutting line of the wedge that is either vertical or horizontal which facilitates the moves of the agent.

\vspace*{0.5cm}

{\bf Boundaries in adjacent quadrants}

\vspace*{0.5cm}

Without loss of generality we assume that both boundaries are North of the origin $O$, $H_1$ has a negative and $H_2$ a positive slope. Hence the wedge has a horizontal cutting line. The algorithm for the other cases is similar. Let $\alpha$ be the acute angle between the horizontal line and the boundary $H_1$, and let $\beta$ be the acute angle between the horizontal line and the boundary $H_2$. Let $x$ be the smallest positive integer such that $1/x \leq \tan \alpha$ and $1/x \leq \tan \beta$. Let $O_1$ and $O_2$ be the nodes of the grid at grid distance $x+1$ from $O$, such that $O_1$ is one step North and $x$ steps West from $O$, and $O_2$ is one step North and $x$ steps East from $O$. Let $H'_1$ be the line passing through $O$ and $O_1$, and let $H'_2$ be the line passing through $O$ and $O_2$ (see Fig. \ref{small free wedge}). Hence the wedge $W'$ with boundaries $H'_1$ and $H'_2$ contains the original wedge $W$. We will show how to explore the wedge $W'$.

\begin{figure}[h]
   \vspace*{-.194in}
     \centering
    \includegraphics[scale = .6]{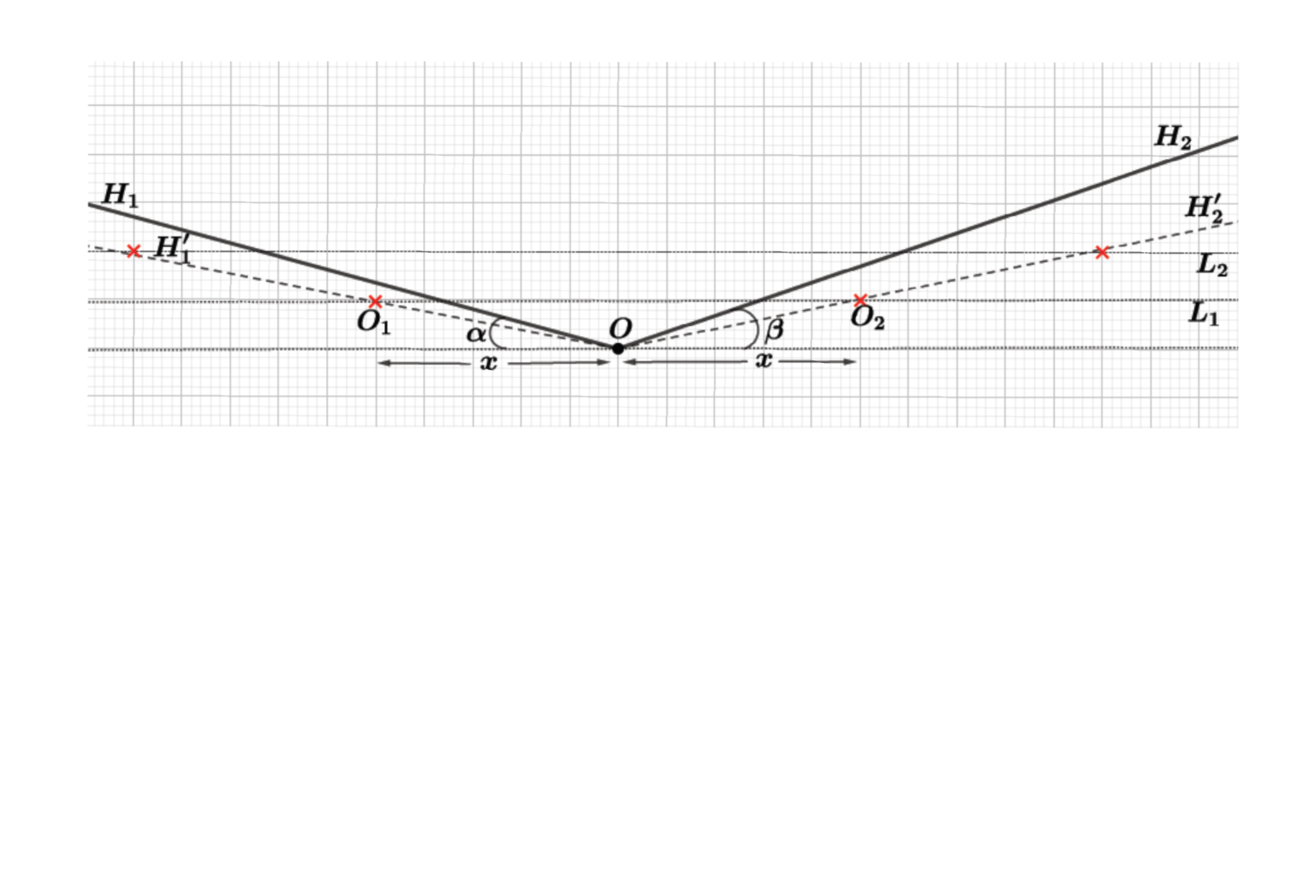}
   \vspace*{-1.6in}
    \caption{Small free wedge with boundaries in adjacent quadrants}
    \label{small free wedge}
\end{figure}

Let $L_i$, for $i\geq 1$, be consecutive horizontal lines at distance $i$ from $O$.
The high-level idea of  Algorithm {\tt Explore Small Free -- Adjacent Quadrants} is to first place the pebbles on line $L_1$ at nodes $O_1$ and $O_2$, and then, in consecutive phases $i>1$, place them on line $L_i$ in such a way that the current location of the left pebble is one step North and $x$ steps West from its previous location, and the current location of the right pebble is one step North and $x$ steps East from its previous location.
Thus the left pebble is on line $H'_1$ and the right pebble is on line $H'_2$. This is done in such a way that in each phase the horizontal segment between the two pebbles is explored.

The algorithm is formulated as follows.

\vspace*{0.5cm}

{\bf Algorithm} {\tt Explore Small Free -- Adjacent Quadrants}

\hspace*{1cm} take port $N$; 

\hspace*{1cm} {\bf for} $i:=1$ {\bf to} $x$ {\bf do} 

\hspace*{2cm} take port $W$; 

\hspace*{1cm} drop a pebble;

\hspace*{1cm} {\bf for} $i:=1$ {\bf to} $2x$ {\bf do} 

\hspace*{2cm} take port $E$; 

\hspace*{1cm} drop a pebble; 

\hspace*{1cm} {\bf repeat forever}

\hspace*{2cm} pick a pebble; 

\hspace*{2cm} take port $N$; 

\hspace*{2cm} {\bf for} $i:=1$ {\bf to} $x$ {\bf do} 

\hspace*{3cm} take port $E$; 

\hspace*{2cm} drop a pebble; 

\hspace*{2cm} take port $S$; 

\hspace*{2cm} {\bf while} no pebble at current node {\bf do} 

\hspace*{3cm} take port $W$;

\hspace*{2cm} pick a pebble; 

\hspace*{2cm} take port $N$; 

\hspace*{2cm} {\bf for} $i:=1$ {\bf to} $x$ {\bf do} 

\hspace*{3cm} take port $W$; 

\hspace*{2cm} drop a pebble;

\hspace*{2cm} take port $E$;

\hspace*{2cm} {\bf while} no pebble at current node {\bf do} 

\hspace*{3cm} take port $E$;

 \vspace*{0.5cm}
 
 \begin{lemma}\label{free1}
 For any small free wedge with boundaries in adjacent quadrants there exists an automaton that explores it with two pebbles.
 \end{lemma}
 
 \begin{proof}
 Before the first execution of the ``repeat forever'' loop the segment of $L_1$ between the intersection points of it with lines $H'_1$ and $H'_2$ is explored. During the $i$-th execution of this loop, the 
 segment of $L_{i+1}$ between the intersection points of it with lines $H'_1$ and $H'_2$ is explored. Since the wedge $W'$ is the union of these segments, the entire wedge $W'$ is eventually explored. Since the integer $x$ can be precomputed knowing the slopes of $H_1$ and $H_2$, the algorithm can be executed by an automaton. Since only two pebbles are used, this proves the lemma.
 \end{proof}
 
 \vspace*{0.5cm}

{\bf Boundaries in opposite quadrants}

 \vspace*{0.5cm}

Now we consider the more difficult case when the boundaries are in opposite quadrants (we include the limit case when one of the boundaries is either vertical or horizontal).
Consider a wedge with boundaries $H_1$ and $H_2$. Without loss of generality assume that the vector determining $H_1$ has both components negative and the vector determining $H_2$ has both components positive. The algorithm for other cases is similar. 

The high-level idea of the algorithm is the following. Let $O^*$ be the node at grid distance 2 from $O$, North-West from it. Let $L_0$ be a rational cutting line of the wedge, passing through $O^*$.
Suppose that it is a $(x,y)$-line. Let $L_i$, for $i\geq 1$, be the line parallel to $ L_0$, one step West of $L_{i-1}$. Let  $T_i$, for $i\geq 0$, be the triangle bounded by the boundaries of the wedge and the line $L_i$. The triangle $T_0$ is uniquely determined and can be precomputed. 

First the agent goes to $O^*$ and explores all grid nodes of $T_0$, getting back to $O^*$.  Then it uses a staircase of $L_0$ to put one pebble on $L_0$, East of or on the boundary $H_2$, and to put the second pebble 
on $L_0$, South of or on the boundary $H_1$. Next the agent works in phases $i=1,2,\dots$ keeping the following invariant: at the beginning of phase $i$, all grid nodes of the triangle $T_{i-1}$ are explored, the pebbles are on line $L_{i-1}$, respectively East of or on the boundary $H_2$ and South of or on the boundary $H_1$. Moreover, the agent is at the lower pebble. During phase $i\geq 1$ the agent uses a staircase of line $L_{i-1}$ to reach the upper pebble, it picks the pebble and goes to a precomputed node $(c',d')$ which is on line $L_{i}$,  East of or on the boundary $H_2$. 
%On its way to $(c',d')$, the agent explores boxes of line $L_{i-1}$ up to the horizontal level of $(c',d')$.
Then the agent drops the pebble at $(c',d')$, goes to the previous location of the pebble, explores the chain of boxes of $L_{i-1}$ and  reaches the lower pebble.  Finally, it moves the lower pebble to a new location $(c'',d'')$ on line $L_i$, South of or on the boundary $H_1$. 
%On its way to $(c'',d'')$, the agent explores boxes of line $L_{i-1}$ down to the horizontal level of $(c'',d'')$.

In order to implement the above idea we define a number of procedures. 
The aim of the first procedure is to get to node $O^*$ and then explore all grid nodes of the triangle $T_0$. Since this triangle is known, this can be done using a simplified version of procedure {\tt Bottom Triangle} from the previous section. The procedure starts at node $O$ and ends at node $O^*$. Let $\Sigma$ be the set of grid nodes in $T_0$.

 \vspace*{0.5cm}
 
 {\bf Procedure} {\tt Explore Free Bottom Triangle}

\hspace*{1cm} take port $N$;

\hspace*{1cm} take port $W$; 

\hspace*{1cm} let $O^*$ be the current node;

\hspace*{1cm} {\tt Visit} $(O^*, \Sigma)$;

\vspace*{0.5cm}

The aim of the next procedure is to place the pebbles at their initial positions on line $L_0$, one pebble East of or on the boundary $H_2$, and the second pebble 
South of or on the boundary $H_1$. The procedure starts at $O^*$ and ends at the position of the second (lower) pebble. Let $z$ be the length of the segment of $L_0$ between the intersection points with boundaries of the wedge. Let $r=\lceil z/y\rceil$. Hence $r$ steps of the staircase up from $O^*$ will get the agent to a node on $L_0$, East of or on the boundary $H_2$, and then $2r$ steps of the staircase of $L_0$ down from that node will get the agent to a node on $L_0$, South of or on the boundary $H_1$.

\vspace*{0.5cm}

{\bf Procedure} {\tt Initialize Both Pebbles}

\hspace*{1cm} {\bf for} $j:=1$ {\bf to} $r$ {\bf do}  

\hspace*{2cm} {\bf for} $i:=1$ {\bf to} $y$ {\bf do} 

\hspace*{3cm} take port $N$;

\hspace*{2cm} {\bf for} $l:=1$ {\bf to} $x$ {\bf do}  

\hspace*{3cm} take port $E$;

\hspace*{1cm} drop a pebble;

\hspace*{1cm} {\bf for} $j:=1$ {\bf to} $2r$ {\bf do}  

\hspace*{2cm} {\bf for} $i:=1$ {\bf to} $x$ {\bf do} 

\hspace*{3cm} take port $W$;

\hspace*{2cm} {\bf for} $l:=1$ {\bf to} $y$ {\bf do}  

\hspace*{3cm} take port $S$;

\hspace*{1cm} drop a pebble;

\vspace*{0.5cm}

The next five procedures are executed one after another in each phase. The aim of the first procedure is to move the agent to the upper pebble.

\vspace*{0.5cm}

{\bf Procedure} {\tt Staircase  Up to Pebble}

\hspace*{1cm} {\bf for} $i:=1$ {\bf to} $y$ {\bf do} 

\hspace*{2cm} take port $N$;

\hspace*{1cm} {\bf for} $j:=1$ {\bf to} $x$ {\bf do}  

\hspace*{2cm} take port $E$;

\hspace*{1cm} {\bf while} no pebble on current node {\bf do}

\hspace*{2cm} {\bf for} $i:=1$ {\bf to} $y$ {\bf do} 

\hspace*{3cm} take port $N$;

\hspace*{2cm} {\bf for} $j:=1$ {\bf to} $x$ {\bf do}  

\hspace*{3cm} take port $E$;

\vspace*{0.5cm}

The next procedure finds a new place for the upper pebble, on line $L_i$ (if the current phase is $i$) East of or on the boundary $H_2$.
%On its way to the new location of the pebble, the agent explores boxes of line $L_{i-1}$ up to the horizontal level of this location.
It is exactly the procedure {\tt Move Pebble} from the previous section. We now call it procedure {\tt Move Upper Pebble}. At the beginning of it the agent picks the (upper) pebble from the current node, and at the end of it the agent drops this  pebble at the new location. 

The aim of the next procedure is to go back to the node where procedure {\tt Move Upper Pebble} started. The agent does this by reversing the path followed during procedure {\tt Move Upper Pebble}.
We call this procedure {\tt Go Back to Previous Location}.

The aim of the next procedure is to explore the chain of boxes between the current node and the lower pebble. The procedure ends at the lower pebble.

\vspace*{0.5cm}

{\bf Procedure} {\tt Explore Chain to Pebble} 

\hspace*{1cm} {\tt Explore Box}

\hspace*{1cm} {\bf while} no pebble on current node {\bf do} 

\hspace*{2cm} {\tt Explore Box} 

\vspace*{0.5cm}

The aim of the final procedure is to find a new place for the lower pebble, on line $L_i$ (if the current phase is $i$) South of or on the boundary $H_1$.
%On its way to the new location of the pebble, the agent explores boxes of line $L_{i-1}$ up to the horizontal level of this location.
This procedure is ``symmetric'' to the procedure {\tt Move Upper Pebble}. We  call it procedure {\tt Move Lower Pebble} and we skip the details. At the beginning of it the agent picks the (lower) pebble from the current node, and at the end it drops it at the new location. This ends phase $i$ and the next phase $i+1$ starts.

Now the algorithm can be succinctly formulated as follows.

\vspace*{0.5cm}

{\bf Algorithm} {\tt Explore Small Free -- Opposite Quadrants}

\hspace*{1cm} {\tt Explore Free Bottom Triangle}

\hspace*{1cm} {\tt Initialize Both Pebbles}

\hspace*{1cm} {\bf repeat forever}

\hspace*{2cm} {\tt Staircase  Up to Pebble}

\hspace*{2cm} {\tt Move  Upper Pebble}

\hspace*{2cm} {\tt Go Back to Previous Location}

\hspace*{2cm} {\tt Explore Chain to Pebble} 

\hspace*{2cm} {\tt Move Lower Pebble}

\begin{lemma}\label{free2}
For any small free wedge with boundaries in opposite quadrants there exists an automaton that explores it with two pebbles.
\end{lemma}

\begin{figure}[h]
   \vspace*{-.194in}
     \centering
    \includegraphics[scale = .5]{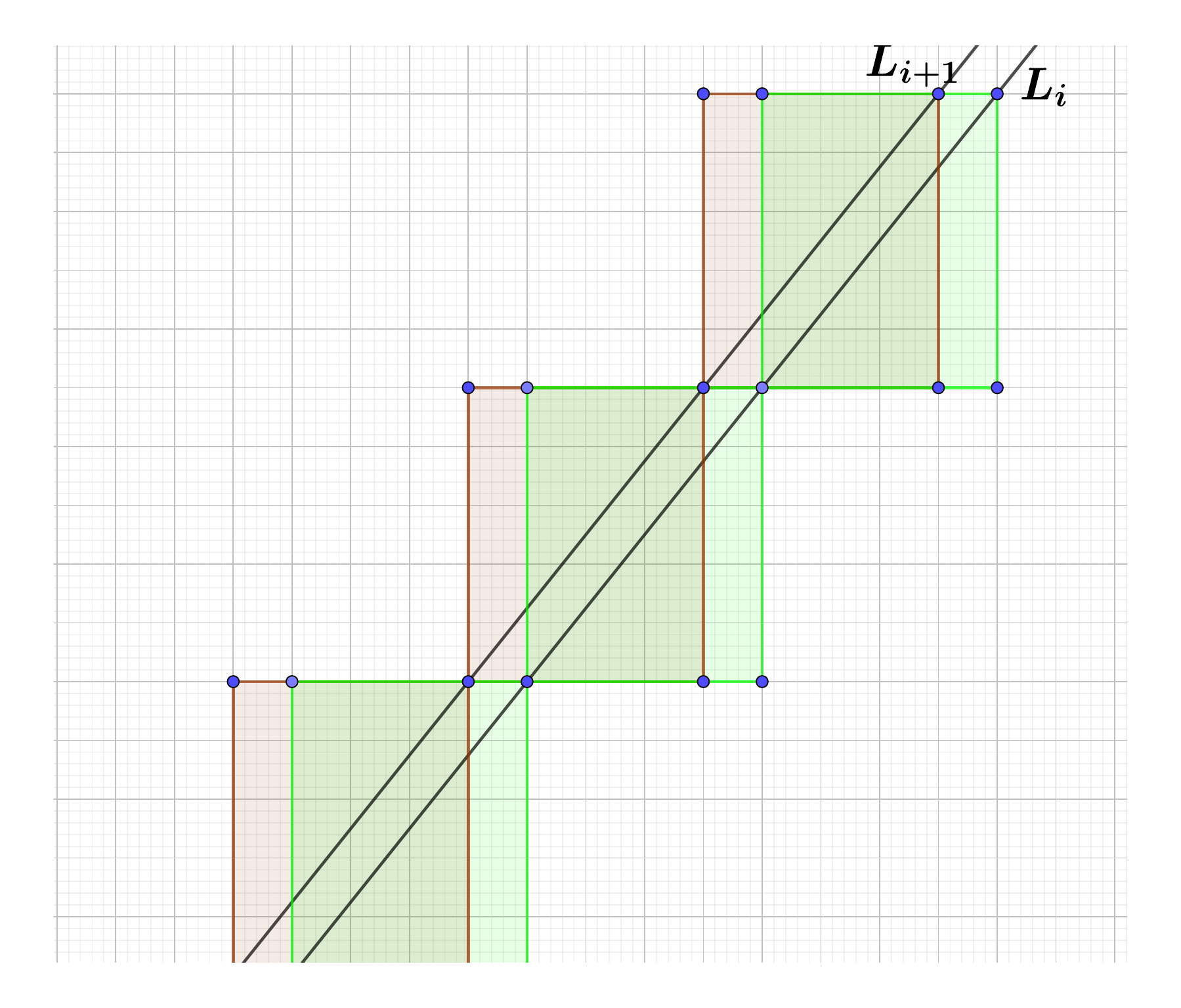}
   \vspace*{-.2in}
    \caption{The chain of boxes of line $L_i$ is explored during the $(i+1)$-th execution of the ``repeat forever'' loop, and the chain of boxes of line $L_{i+1}$ is explored during the $(i+2)$-th execution of the ``repeat forever'' loop}
    \label{free}
\end{figure}

\begin{proof}
Let $Inv(i)$, for $i\geq 2$, be the following invariant: after the $i$-th execution of the ``repeat forever'' loop, all grid nodes in the triangle $T_{i-1}$ are explored. 
By the description of procedure {\tt Explore Free Bottom Triangle}, all grid nodes in $T_0$ are visited before the first execution of the ``repeat forever'' loop.
During the first execution of this loop the chain of boxes of $L_0$ is explored and during the second execution of this loop the chain of boxes of $L_1$ is explored. Hence all grid nodes of the wedge
between lines $L_0$ and $L_1$ are explored by the end of the second execution of the loop. Hence
$Inv(2)$ holds. 
Assuming that $Inv(i)$ holds, the descriptions of the five procedures executed in each turn of the loop implies that $Inv(i+1)$ holds.
Hence  $Inv(i)$ holds for all $i\geq 2$, by induction.
Since the wedge is the union of all sets of grid nodes in triangles $T_i$, for $i\geq 2$, we conclude that Algorithm {\tt Explore Small Free -- Opposite Quadrants} explores the entire wedge.
Since all integer parameters used in the procedures called by the algorithm can be precomputed knowing the slopes of $H_1$ and $H_2$, the algorithm can be executed by an automaton. Since only two pebbles are used, this proves the lemma.
\end{proof}

Lemmas \ref{free1} and \ref{free2} imply the main positive result of this section.

\begin{theorem}\label{th-small-free}
For any small free wedge there exists an automaton that explores it with two pebbles.
\end{theorem}

\subsubsection{Impossibility of exploration with one pebble}

The next theorem shows that the number of pebbles in Theorem \ref{th-small-free} cannot be decreased. The theorem follows from the proof of Theorem 4 in \cite{ELSUW} where the authors show that 
two automata in a grid can only explore nodes in a band of bounded width. Since an automaton can simulate a pebble and no wedge is contained in such a band, this proves the following result.

\begin{theorem}\label{lb small free}
Let $W$ be any free wedge. Then $W$ cannot be explored by an automaton with one pebble.
\end{theorem}

\begin{corollary}
The minimum number of pebbles sufficient to explore any small free wedge is two.
\end{corollary}

\section{Large wedges}

In this section we assume that the clockwise angle between boundaries $H_1$ and $H_2$ of the wedge is $\alpha \geq \pi$.

\subsection{Walled wedges}

\subsubsection{Exploration with one pebble}

In this section we show that for any large walled wedge there exists an automaton that explores it with one pebble.
Without loss of generality assume that the vectors determining both boundaries have both components positive. The algorithm for the other cases, including the limit cases when boundaries are horizontal or vertical, is similar.

The high-level idea of the algorithm is the following. We choose a rational $(x,y)$-line $L$ containing the initial node of the agent such that the slope of this line is between the slopes of $H_1$ and $H_2$. 
Hence, for any node $v$ of the wedge, the line parallel to $L$ and containing $v$ intersects one of the boundaries $H_1$ or $H_2$. First the agent hits one of the walls moving along a staircase of $L$. Then it slides down to a node $O^*$ close to the origin $O$, and places the pebble at this node. The pebble will be moved along the line $L^*$ parallel to $L$ and passing through $O^*$, using a staircase of this line.
Two other rational lines $L_1$ and $L_2$ are chosen in such a way that the slope of $L_1$ is between the slopes of $H_1$ and $L^*$, and the slope of $L_2$ is between the slopes of $H_2$ and $L^*$.
Let $L_1$ be an $(x_1,y_1)$-line and let $L_2$ be an $(x_2,y_2)$-line. By definition, for any position $P$ of the pebble on line $L^*$ below the node $O^*$, the line passing through $P$ and parallel to $L_1$ must intersect the wall $H_1$, and the line passing through $w$ and parallel to $L_2$ must intersect the wall $H_2$ (see Fig. \ref{region}).

\begin{figure}[h]
   \vspace*{-.194in}
     \centering
    \includegraphics[scale = .4]{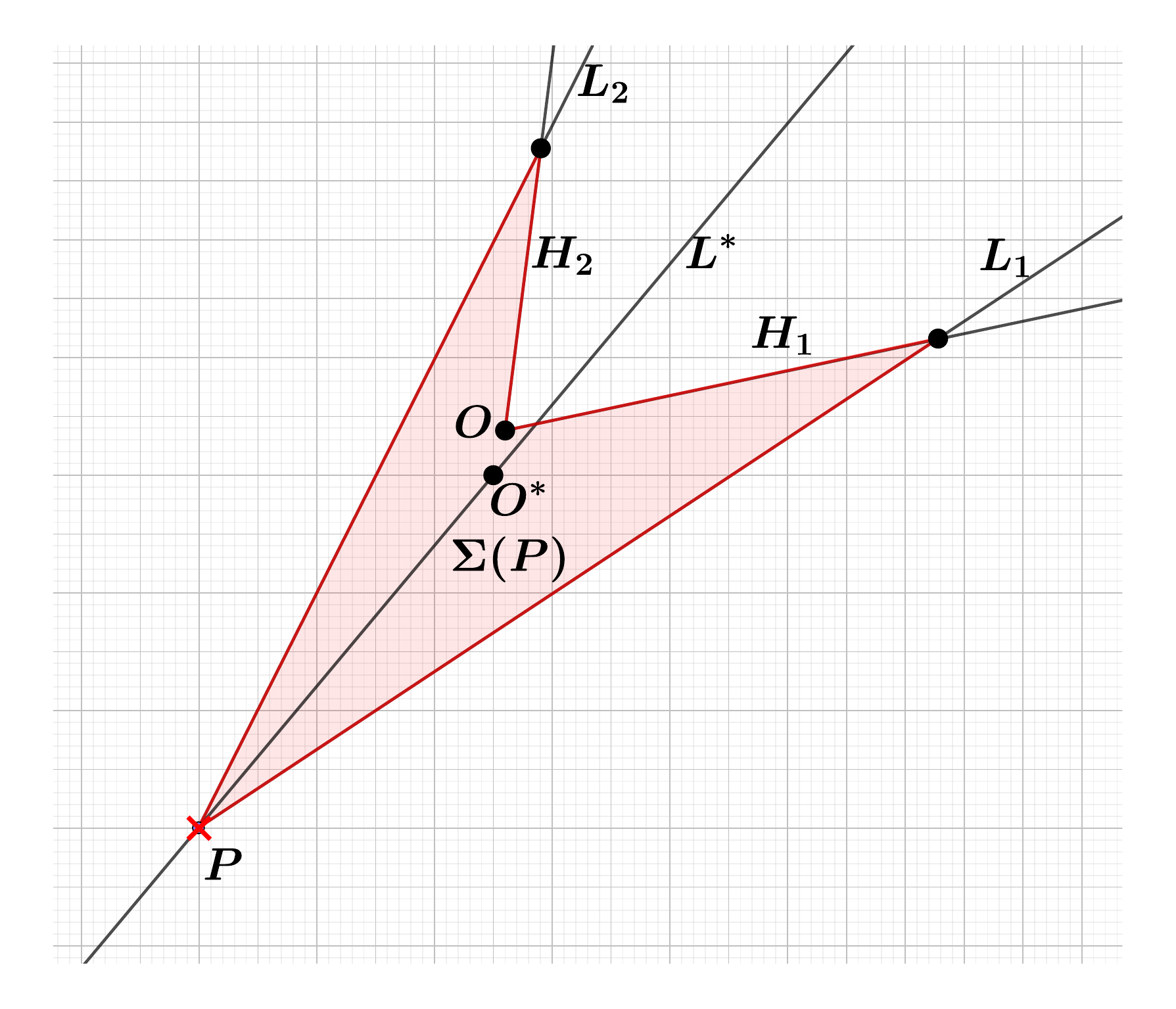}
   \vspace*{-.2in}
    \caption{The region $\Sigma(P)$ between lines $L_1$, $L_2$ and the walls $H_1$ and $H_2$ of the wedge. $P$ is the current position of the pebble.}
    \label{region}
\end{figure}

The rest of the algorithm works in phases executed by a ``repeat forever'' loop. In each phase, 
%the pebble is moved one  $(x,y)$-step down along a staircase of $L^*$. Then 
the agent starts at the pebble, goes up a staircase of line $L_2$ to hit  the wall $H_2$ and explores a chain of boxes of this line to return to the pebble. Then a similar procedure is executed along line $L_1$, the agent gets back to the pebble which is then moved one  $(x,y)$-step down along a staircase of $L^*$ and a new phase starts.
We keep the invariant that after each phase, the part of the wedge between the walls and the lines parallel to $L_1$ and $L_2$ passing through the previous position of the pebble is explored.

In order to implement the above idea we formulate several procedures.
The first of them starts at the starting node of the agent and aims at hitting a wall. It uses a staircase of line $L$. Upon its completion, the agent may either hit the right wall, i.e., boundary $H_1$ (which is indicated by the variable $right-wall$ becoming true, or it may hit the left wall, i .e., boundary $H_2$ (which is indicated by the variable $left-wall$ becoming true). The agent can infer which of the walls it hits
by seeing which port is blocked.

\vspace*{0.5cm}

{\bf Procedure} {\tt Staircase along $L$ up to Wall}

\hspace*{1cm} $hit :=false$;

\hspace*{1cm} {\bf while} $hit=false$ {\bf do}

\hspace*{2cm} $i:=0$;

\hspace*{2cm} {\bf while} $N$ is free {\bf and} $i<y$ {\bf do}

\hspace*{3cm} take port $N$;

\hspace*{3cm} $i:=i+1$;

\hspace*{2cm} {\bf if} $N$ is blocked {\bf then}

\hspace*{3cm} $hit:=true$;

\hspace*{3cm} $right-wall:=true$;

\hspace*{2cm} $j:=0$;

\hspace*{2cm} {\bf while} $hit=false$ {\bf and} $E$ is free {\bf and} $j<x$ {\bf do}

\hspace*{3cm} take port $E$;

\hspace*{3cm} $j:=j+1$;

\hspace*{2cm} {\bf if} $E$ is blocked {\bf then}

\hspace*{3cm} $hit:=true$;

\hspace*{3cm} $left-wall:=true$;

\vspace*{0.5cm}

The next procedure starts where the previous one ended and it has two variations, depending on which wall was hit during the execution of procedure  {\tt Staircase along $L$ up to Wall}. Its aim is to bring the agent
to the horizontal level of the origin $O$ of the wedge, or South of it. It is somewhat similar to procedure {\tt Go to Origin}. We only show the variation when the left wall was hit: the other variation is similar.
The procedure ends at a node $O^*$ which is either the origin $O$ in case when $O$ is a grid node, or otherwise it is the South-West corner of the elementary square of the grid containing $O$. The agent drops the pebble at $O^*$.

\vspace*{0.5cm}

{\bf Procedure} {\tt Slide Down to Origin - Left}

\hspace*{1cm} {\bf while} $E$ is blocked {\bf do}

\hspace*{2cm} {\bf if} $S$ is free {\bf then}

\hspace*{3cm} take port $S$;

\hspace*{2cm} {\bf else}

\hspace*{3cm} take port $W$;

\hspace*{3cm} take port $S$;

\hspace*{1cm} drop a pebble;

\vspace*{0.5cm}

We will need two other procedures:  {\tt Staircase along $L_1$ up to Wall} and  {\tt Staircase along $L_2$ up to Wall}. Their aim is to start at the current position $p$ of the pebble and go up along a staircase
of the line $L'_1$ passing through $p$ and parallel to $L_1$ (resp. of the line $L'_2$ passing through $p$ and parallel to $L_2$) until hitting wall $H_1$ (resp. wall $H_2$). Procedure {\tt Staircase along $L_2$ up to Wall} aims at hitting the left wall $H_2$ and procedure {\tt Staircase along $L_1$ up to Wall} aims at hitting the right wall $H_1$. Procedure {\tt Staircase along $L_2$ up to Wall}  is obtained from procedure 
{\tt Staircase along $L$ up to Wall} by the following changes: remove variables $right-wall$ and $left-wall$ and replace $x$ by $x_2$ and $y$ by $y_2$ to reflect the change of line $L$ to line $L_2$.
Procedure {\tt Staircase along $L_1$ up to Wall} is symmetric to procedure  {\tt Staircase along $L_2$ up to Wall} and uses $x_1$ instead of $x_2$ and $y_1$ instead of $y_2$, to reflect the change to line $L_1$.

The next two procedures explore the accessible part of a box of side sizes $x\cdot x_i$ and $x\cdot y_i$ of a line $L'_i$, for $i=1,2$, parallel to $L_i$. They both start at the North-East corner of a box, explore the part of it contained in the wedge and end at the South-West corner of the box. We only describe the variation for line $L_2$ (i.e., the left variation). The variation for line $L_1$ is similar. The reason to explore such large boxes, instead of boxes of side sizes $x_i$ and $y_i$, is to guarantee that all grid nodes between consecutive lines  parallel to $L_i$  are visited in two consecutive phases. These consecutive lines are at horizontal distance smaller than $x$ from each other, due to the moves of the pebble along $L^*$.

\vspace*{0.5cm}

{\bf Procedure} {\tt Conditional Box Exploration - Left}

\hspace*{1cm} {\bf for} $i:=1$ {\bf to} $x\cdot x_2$ {\bf do}

\hspace*{2cm} take port $W$;

\hspace*{1cm} {\bf for} $j=1$ {\bf to} $x\cdot y_2$ {\bf do}

\hspace*{2cm} take port $S$;

\hspace*{2cm} $k:=0$;

\hspace*{2cm} {\bf while} $E$ is free {\bf and} $k<x\cdot x_2$ {\bf do}

\hspace*{3cm} Take port $E$;

\hspace*{3cm} $k:=k+1$;

\hspace*{2cm} {\bf for} $m:=1$ {\bf to} $k$ {\bf do} 

\hspace*{3cm} take port $W$;

\vspace*{0.5cm}

The next two procedures explore a chain of boxes and  end at the current position of the pebble. We only formulate the left variation. The right variation is similar.

\vspace*{0.5cm}

{\bf Procedure} {\tt Box Chain to Pebble - Left}

\hspace*{1cm} {\bf while} no pebble at current node {\bf do}

\hspace*{2cm} {\tt Conditional Box Exploration - Left}

\vspace*{0.5cm}

Finally we formulate the procedure that moves the pebble down along $L^*$.

\vspace*{0.5cm}

{\bf Procedure} {\tt Move Pebble}
 
 \hspace*{1cm} pick a pebble;
 
  \hspace*{1cm} {\bf for} $i:=1$ {\bf to} $x$ {\bf do}
  
  \hspace*{2cm} take port $W$;

 \hspace*{1cm} {\bf for} $j:=1$ {\bf to} $y$ {\bf do}
  
 \hspace*{2cm} take port $S$;
  
 \hspace*{1cm} drop a pebble;

 \vspace*{0.5cm}
 
 We now formulate our algorithm exploring large walled wedges using one pebble.
 
  \vspace*{0.5cm}
  
  {\bf Algorithm} {\tt Explore Large Walled} 

\hspace*{1cm} {\tt Staircase along $L$ up to Wall}

\hspace*{1cm} {\bf if} $left-wall=true$ {\bf then}

\hspace*{2cm} {\tt Slide Down to Origin - Left}

\hspace*{1cm} {\bf else}

\hspace*{2cm} {\tt Slide Down to Origin - Right}

\hspace*{1cm} {\bf repeat forever}

\hspace*{2cm} {\tt Staircase along $L_2$ up to Wall}

\hspace*{2cm} {\tt Box Chain to Pebble - Left}

\hspace*{2cm} {\tt Staircase along $L_1$ up to Wall} 

\hspace*{2cm} {\tt Box Chain to Pebble - Right} 

\hspace*{2cm} {\tt Move Pebble}

\vspace*{0.5cm}

The following is the main positive result of this section.

\begin{theorem}\label{th large walled}
For any large walled wedge there exists an automaton that explores it with one pebble.
\end{theorem}

\begin{proof}
For any position $P$ of the pebble on line $L^*$ denote by $\Sigma(P)$ the region between the two walls and lines $L'_1$ passing through $P$ and parallel to $L_1$ and  $L'_2$ passing through $P$ and parallel to $L_2$ (see Fig. \ref{region}).
When the pebble is first placed at node $O^*$ then the set of grid nodes of the region $\Sigma(O^*)$ consists of the node  $O^*$ which is visited. 
Let phase $i$, for $i=1,2,\dots$, be the part of the algorithm executed during the $i$-th turn of the ``repeat forever'' loop. As the union of all regions $\Sigma(P)$, for all positions $P$ of the pebble, contains all nodes of the wedge, it is enough to show that after phase $i$ all nodes of $\Sigma(P')$ where $P'$ is the previous position of the pebble, are visited. (Recall that the move of the pebble is at the very end of each phase).

In order to prove this, define $P_i$ to be the position of the pebble before its move in the $i$-th phase. Define the line $A_i$ to be the line parallel to $L_2$ passing through $P_i$, and define the line $B_i $ to be the line parallel to $L_1$ passing through $P_i$. Denote by $R_i$ the region between lines $L^*$, $A_{i-1}$, $A_i$ and $H_2$ (see Fig. \ref{region R_i}).

\begin{figure}[h]
   \vspace*{-.194in}
     \centering
    \includegraphics[scale = .5]{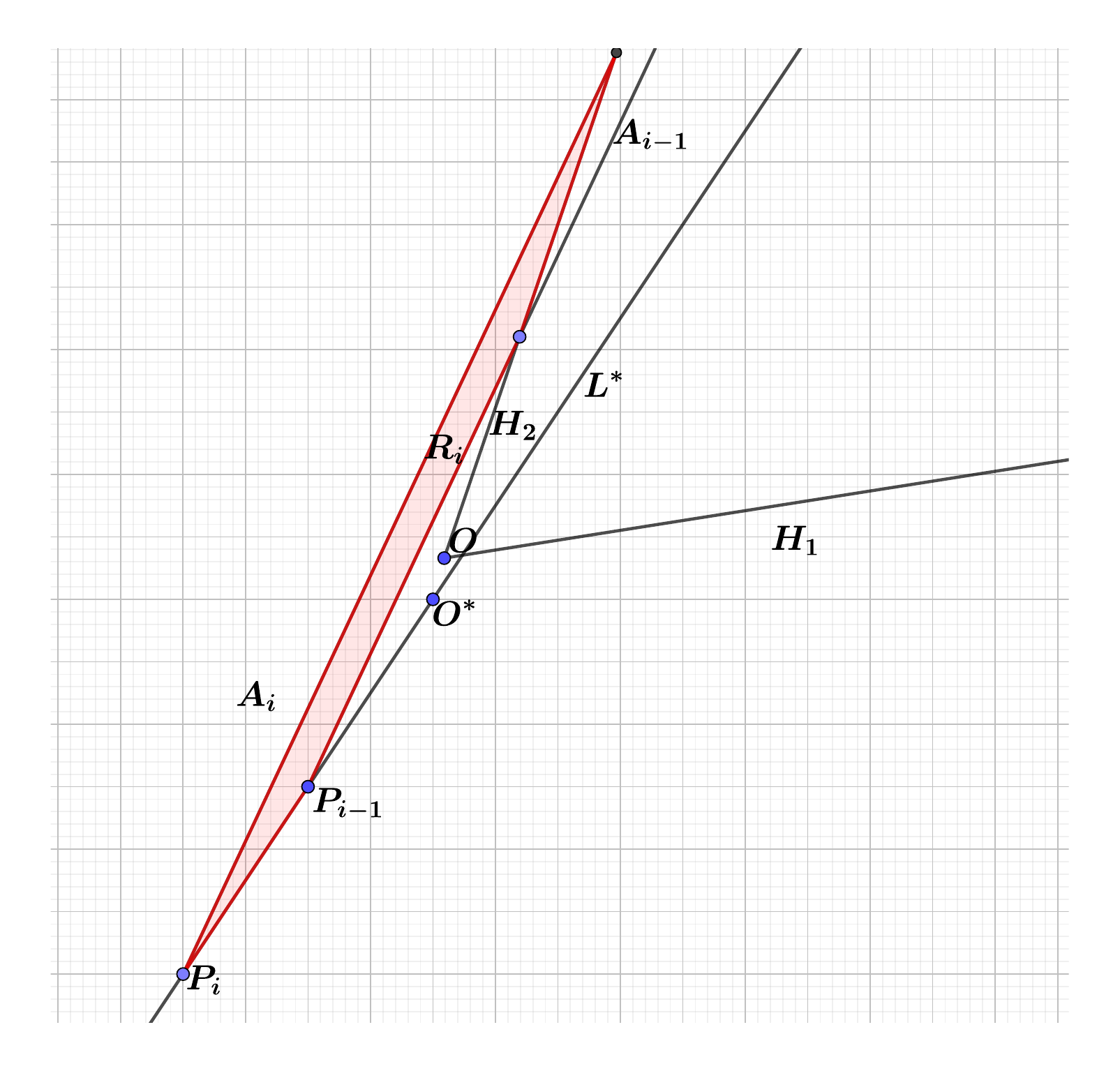}
   \vspace*{-.2in}
    \caption{The region $R_i$}
    \label{region R_i}
\end{figure}

Similarly, denote by $S_i$ the region between lines $L^*$, $B_{i-1}$, $B_i$ and $H_1$. Since the horizontal distance between lines $A_{i-1}$, and $A_i$ is less than $x$, and the width of boxes used in procedure {\tt Box Chain to Pebble - Left} is larger than $x$, it follows that during the execution of phases $i-1$ and $i$, all grid nodes of the region $R_i$ are visited. Similarly, during the execution of phases $i-1$ and $i$, all grid nodes of the region $S_i$ are visited. Hence, after phase $i$, all nodes of $\Sigma(P')$, where $P'$ is the previous position of the pebble, are visited, as desired.
\end{proof}

\subsubsection{Impossibility of exploration without a pebble}

In this section we show that the number of pebbles in Theorem \ref{th large walled} cannot be decreased.

\begin{theorem}
No large walled wedge can be explored by an automaton without any pebble.
\end{theorem}

\begin{proof}
%Any large wedge can be partitioned into a sequence of pairwise disjoint {\em layers} $\lambda_i$, for $i=1,2,\dots$, with the following properties:
%
%\begin{itemize}
%\item
%each layer is a simple two-way infinite path in the grid, i.e., it can be represented as $(\dots, v_{-2}, v_{-1}, v_0, v_1, v_2, \dots)$, where all the nodes $v_i$ are different and $v_i$ and $v_{i+1}$ are adjacent;
%\item
%all ports at any node of any layer $\lambda_i$, for $i\geq 2$, are free;
%\item
%if $v$ is a node of layer $\lambda_i$ and $w$ is a node of $\lambda_{i+2}$ then any path in the wedge from $v$ to $w$ must contain a node from layer $\lambda_{i+1}$;
%\item
%if $H$ is a half-line starting at some node $v$ of layer $\lambda_i$ and contains some node $w$ of a layer $\lambda_{j}$, for $j>i$, then $H$ does not intersect any of the boundaries of the wedge.
%\end{itemize}
%
%

Consider any large walled wedge $W$ with walls $H_1$ and $H_2$, and any automaton with $s$ states exploring it without any pebble.
Let $T$ be the (infinite) trajectory of this automaton starting at any node of $W$.
Let $P_i$, for $i=1,2,\dots$, be any family of parallel lines included in the part of the plane clockwise between $H_1$ and $H_2$, such that $P_i$ is at distance $i$ from the origin $O$ of the wedge.
Let $B_i$, for $i=1,2,\dots$, be the set of grid nodes situated between lines $P_i$ and $P_{i+1}$, including $P_i$ and excluding $P_{i+1}$.  Call any $B_i$ a {\em belt}. 
Consider the part $T'$ of the trajectory $T$ after the first visit of $O$. 

\vspace*{0.5cm}

\noindent
{\bf Claim.}
No grid node in any belt $B_i$, for $i>s$, can be visited during trajectory $T'$.

We prove the claim by contradiction. 
Suppose that the agent visits node $O$ for the first time in some step $t$, and then reaches the belt $B_{i}$, for $i>s$. 
Then there exist nodes $w_i$ in $B_i$ visited in step $t_i$, for any $i\geq 1$, such that $t_j>t_i$, for $j>i$, and the agent never visits a node of $B_1$ after step $t_1$.
It follows that there are grid nodes $v'=w_i$  and $v''=w_j$ for some $j>i$, such that the agent is in the same state $S$ when visiting node $v'$ in step $t'=t_i$ and node $v''$ in step $t''=t_j$.

 Let $v_0,v_1,\dots, v_r$, with $v_0=v'$ and $v_r=v''$,  be the sequence of nodes visited between steps $t'$ and $t''$. Hence $t''=t'+r$.  Let $\alpha$ be the vector $(v',v'')$. Hence in step $t'+kr+j$, for any natural $k$ and any $j<r$, the automaton will be in the same state $S_j$ at node $v_j+k\alpha$. (Intuitively, the trajectory $(v_0,v_1,\dots, v_r)$ is shifted infinitely many times, from step $t'$ on, by multiples of the vector $\alpha$). Let $L_1$ and $L_2$ be two lines parallel to vector $\alpha$, such that all nodes $v_0,v_1,\dots, v_r$ are between these lines. It follows that all nodes of the part of the trajectory $T$ after step $t'$ are between lines $L_1$ and $L_2$. Consequently, at most $t'$ nodes of the trajectory $T$ are not between these lines. Since infinitely many nodes of $W$ are not between lines $L_1$ and $L_2$, there exist nodes of $W$ which will never be visited. This contradiction proves the claim.
 
Let $B$ be the union of all belts $B_i$, for $i>s$. By the claim, the agent cannot visit any grid node of $B$ after step $t$. Since $B$ is an infinite set included in the wedge $W$, and at most $t$ nodes of this set can be visited, it follows that there is a node of $W$ that the agent never visits. This contradiction proves the theorem.
\end{proof}

\subsection{Semi-walled wedges}

In this section we consider large semi-walled wedges. We will show that every large semi-walled wedge can be explored by an automaton with two pebbles and that one pebble is not enough to explore any large semi-walled wedge.

\subsubsection{Exploration with two pebbles}

We first show how to explore any large semi-walled wedge using two pebbles. It is enough to consider the hardest case, when the wedge is the entire grid with one half-line $H$ being the single wall. Without loss of generality assume that the vector determining the wall has both components positive. The algorithm for the other cases is similar.

The high-level idea of Algorithm {\tt Explore Large Semi-Walled} is the following. First use a modified version of Algorithm {\tt Explore Small Free - Adjacent Quadrants} until the agent hits the wall. The small free wedge used in this algorithm is the wedge with boundary $H_1$ defined by the vector $(-1,1)$ and with boundary $H_2$ having positive slope smaller than that of the wall of the input wedge. It is easy to see that, regardless of the initial position of the agent, the part of the plane clockwise between boundaries $H_1$ and $H_2$ must intersect the wall $H$. Since Algorithm {\tt Explore Small Free - Adjacent Quadrants} explores the entire wedge with these boundaries, the agent must hit the wall.

The modified procedure makes all moves of Algorithm {\tt Explore Small Free - Adjacent Quadrants} conditional on the fact that the currently used port is free. Once the agent hits the wall at some node $w$, it learns whether it is below or above the wall, as this is indicated by which port at $w$ is blocked. This is marked in the modified procedure by switching the variable $below$ to $true$ in the first case and leaving it $false$ otherwise. Next the agent picks both pebbles and gets back to node $w$.
We skip the details of this procedure and call it {\tt Modified Small Free Exploration}.

Then the agent ``slides down'' the wall to a node $O^*$ equal or close to $O$, at the horizontal level of $O$ or South of it.  The agent drops the two pebbles one East of  $O^*$ and the other West of $O^*$ and goes to the left pebble.
The rest of the algorithm is executed in phases, where each phase is one turn of a ``repeat forever '' loop. Each phase starts with the agent at the left pebble. We choose a rational $(x,y)$-line $L$ with slope smaller than the slope of the wall $H$. For any position $p$ of the left pebble the line parallel to $L$ and going through $p$ is guaranteed to intersect the wall, and for any position $q$ of the right pebble the vertical line containing $q$ intersects the wall.  In each phase, the agent goes from the left pebble up a staircase of the line parallel to $L$ and passing through the current position $p$ of the left pebble, until hitting the wall. Then it goes back to the left pebble exploring a chain of boxes. Next the agent goes to the right pebble, goes North from it until hitting the wall, and goes back to the right pebble. The phase is finished by moving the pebbles to new positions. The new position $q'$ of the right pebble is South and East from $q$ and the new position $p'$ of the left pebble is  South and West from $p$. While moving the pebbles the agent explores the region between the new and the old positions of the pebbles (see Fig. \ref{phase}). Then a new phase starts with the agent at the new position of the left pebble.

\begin{figure}[h]
   \vspace*{-.194in}
     \centering
    \includegraphics[scale = .5]{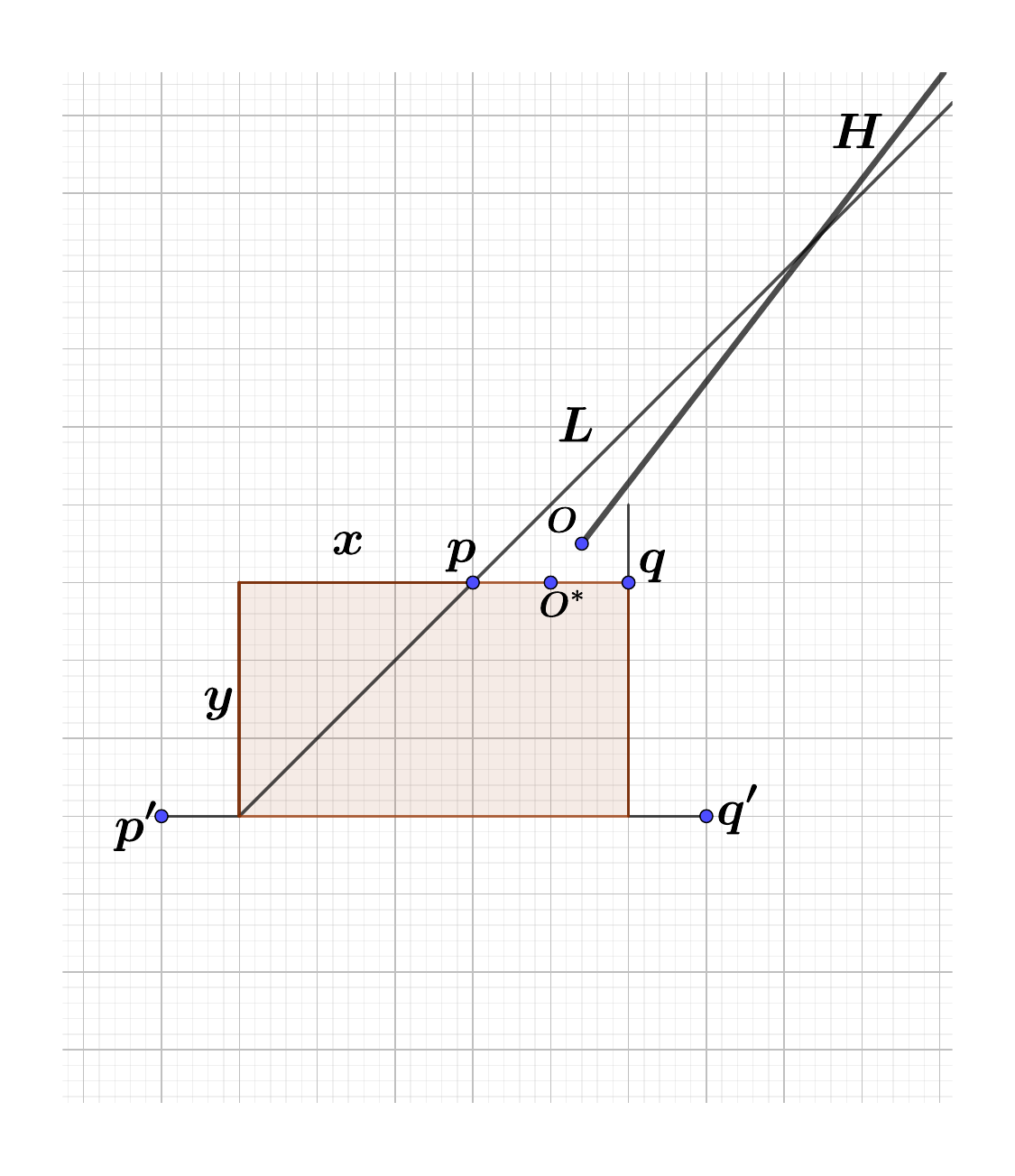}
   \vspace*{-.2in}
    \caption{Illustration of a phase of Algorithm {\tt Explore Large Semi-Walled}. The shaded rectangle is explored while moving the pebbles.  }
    \label{phase}
\end{figure}

In order to implement this idea we define several procedures.
The first of them starts where procedure {\tt Modified Small Free Exploration} ended
and it has two variations:  {\tt Slide Down to Origin - Above} when the wall was hit from above, and  {\tt Slide Down to Origin - Below} when the wall was hit from below.
Its aim is to bring the agent
to the horizontal level of the origin $O$ of the wedge, or South of it. It is somewhat similar to procedure {\tt Go to Origin}. We only show the variation when the wall was hit from above: the other variation is similar.

\vspace*{0.5cm}

{\bf Procedure} {\tt Slide Down to Origin - Above}

\hspace*{1cm} {\bf while} $E$ is blocked {\bf do}

\hspace*{2cm} {\bf if} $S$ is free {\bf then}

\hspace*{3cm} take port $S$;

\hspace*{2cm} {\bf else}

\hspace*{3cm} take port $W$;

\hspace*{3cm} take port $S$;

\vspace*{0.5cm}

The next procedure initializes pebbles in such a way that they are on the same horizontal line at distance 2 from each other, the left one left of $O$ and the right one right of $O$, and they are at the level of $O$ or South of this level. We only describe the variation corresponding to procedure  {\tt Slide Down to Origin - Above}. The other variation is similar. In both cases the agent ends up at the left pebble.

\vspace*{0.5cm}

{\bf Procedure} {\tt Initializing Pebbles - from Above}

\hspace*{1cm} {\tt Modified Small Free Exploration}

\hspace*{1cm} {\tt Slide Down to Origin - Above}

\hspace*{1cm} take port $E$;

\hspace*{1cm} drop a pebble;

\hspace*{1cm} take port $W$;

\hspace*{1cm} take port $W$;

\hspace*{1cm} drop a pebble;

\vspace*{0.5cm}

The next procedure starts at the current position $p$ of the left pebble and uses a staircase along the line parallel to $L$ and passing through $p$, in order to hit the wall. 

\vspace*{0.5cm}

{\bf Procedure} {\tt Staircase up to Wall}

\hspace*{1cm} $hit :=false$;

\hspace*{1cm} {\bf while} $hit=false$ {\bf do}

\hspace*{2cm} $i:=0$;

\hspace*{2cm} {\bf while} $N$ is free {\bf and} $i<y$ {\bf do}

\hspace*{3cm} take port $N$;

\hspace*{3cm} $i:=i+1$;

\hspace*{2cm} {\bf if} $N$ is blocked {\bf then}

\hspace*{3cm} $hit:=true$;

\hspace*{2cm} $j:=0$;

\hspace*{2cm} {\bf while} $hit=false$ {\bf and} $E$ is free {\bf and} $j<x$ {\bf do}

\hspace*{3cm} take port $E$;

\hspace*{3cm} $j:=j+1$;

\hspace*{2cm} {\bf if} $E$ is blocked {\bf then}

\hspace*{3cm} $hit:=true$;

\vspace*{0.5cm}

The next procedure explores a chain of boxes and ends at the current position of the left pebble.
It uses procedure  {\tt Conditional Box Exploration - Left} from the previous section.

\vspace*{0.5cm}

{\bf Procedure} {\tt Box Chain to Left Pebble}

\hspace*{1cm} {\bf while} no pebble at current node {\bf do}

\hspace*{2cm} {\tt Conditional Box Exploration - Left}

\vspace*{0.5cm}

The aim of the next procedure is to find new positions of the pebbles at the end of a phase. Both new positions  have to be South of the old ones, on the same horizontal line.
Suppose that the old position of the left pebble was on line $L'$ parallel to $L$ and the old position of the right pebble was on vertical line $V$.
 The new position of the left pebble has to be on the line $L^*$ parallel to $L$,  one step West of $L'$,  and the new position of the right pebble has to be on the vertical line $V^*$, one step East of $V$. This will ensure widening of the trajectories in consecutive phases. Since moving along the lines parallel to $L$ can be only done using staircases, the new positions of the pebbles will be $y$ steps South of the previous positions. This in turn requires visiting all nodes in a rectangle of height $y$ and width larger by $x$ than the distance between the old positions of pebbles (otherwise grid nodes of this rectangle would be missed). The following procedure starts at the left pebble and accomplishes all this. 
 
 \vspace*{0.5cm}
 
 {\bf Procedure} {\tt Move Pebbles}
 
 \hspace*{1cm} pick a pebble;
 
  \hspace*{1cm} {\bf for} $i:=1$ {\bf to} $x$ {\bf do}
  
  \hspace*{2cm} take port $W$;
  
  \hspace*{1cm} drop a pebble;
  
%  \hspace*{1cm} take port $E$;
%  
%  \hspace*{1cm} {\bf while} no pebble on current node {\bf do}
%  
%  \hspace*{2cm} take port $E$;
%  
%  \hspace*{1cm} pick a pebble;
% 
%  \hspace*{1cm} {\bf for} $i:=1$ {\bf to} $x$ {\bf do}
%  
%  \hspace*{2cm} take port $W$;
%  
%  \hspace*{1cm} drop a pebble;
%  
%  \hspace*{1cm} take port $W$;
%  
%  \hspace*{1cm} {\bf while} no pebble on current node {\bf do}
%  
%  \hspace*{2cm} take port $W$;

 \hspace*{1cm} {\bf for} $j:=1$ {\bf to} $y$ {\bf do}
 
 \hspace*{2cm} pick a pebble;
 
 \hspace*{2cm} take port $S$;
  
 \hspace*{2cm} drop a pebble;
 
 \hspace*{2cm} take port $N$;
 
 \hspace*{2cm} {\bf while} no pebble on current node {\bf do}
 
 \hspace*{3cm} take port $E$; 
 
 \hspace*{2cm} pick a pebble;
 
 \hspace*{2cm} take port $S$;
  
 \hspace*{2cm} drop a pebble;
 
 \hspace*{2cm} take port $W$;
 
 \hspace*{2cm} {\bf while} no pebble on current node {\bf do}

 \hspace*{3cm} take port $W$; 
 
 \hspace*{1cm} pick a pebble;
 
 \hspace*{1cm} take port $W$;
 
 \hspace*{1cm} drop a pebble;
 
 \hspace*{1cm} take port $E$;
 
 \hspace*{1cm} {\bf while} no pebble on current node {\bf do}
 
 \hspace*{2cm} take port $E$;
 
 \hspace*{1cm} pick a pebble;
 
\hspace*{1cm} take port $E$;

\hspace*{1cm} drop a pebble;

\hspace*{1cm} take port $W$;

\hspace*{1cm}  {\bf while} no pebble on current node {\bf do}

\hspace*{2cm} take port $W$;

 \vspace*{0.5cm}
 
 We are now able to formulate our exploration algorithm. In addition to the above formulated procedures, it will use procedure {\tt Slide Down to Origin - Below} which is the counterpart of 
  procedure {\tt Slide Down to Origin - Above},  in the case where the wall was hit from below during the execution of  procedure {\tt Modified Small Free Exploration}. Procedure {\tt Initializing Pebbles - from Below} will be the counterpart of procedure {\tt Initializing Pebbles - from Above} with  procedure {\tt Slide Down to Origin - Below} replacing procedure {\tt Slide Down to Origin - Above}.
 
 In each phase executed by  a turn of the ``repeat forever'' loop, the agent that is at the left pebble uses procedure {\tt Staircase up to Wall}  to hit the wall, then it uses procedure {\tt Box Chain to Left Pebble} to do the left part of the exploration. Upon completion of it, the agent is at the left pebble. Then it moves to the right pebble. Here the situation is easier because the agent may move vertically.
 It uses the following subroutine.
 
  \vspace*{0.5cm}
 
 {\bf Procedure} {\tt Vertical Trip} 
 
 \hspace*{1cm} {\bf while} $N$ is free {\bf do}
 
 \hspace*{2cm} take port $N$;
 
 \hspace*{1cm} {\bf while} no pebble on current node {\bf do}
 
 \hspace*{2cm} take port $S$;
 
 \vspace*{0.5cm}

 After completion of it the agent is back at the right pebble. 
Finally, it goes to the left pebble and it uses procedure {\tt Move Pebbles} to place the pebbles on the new positions and to perform the bottom part of the exploration.
One phase of the algorithm is formulated as follows.

\vspace*{0.5cm}

{\bf Procedure} {\tt Phase}

\hspace*{1cm} {\tt Staircase up to Wall} 

\hspace*{1cm} {\tt Box Chain to Left Pebble}

\hspace*{1cm} take port $E$;

\hspace*{1cm}  {\bf while} no pebble on current node {\bf do}

\hspace*{2cm} take port $E$;

\hspace*{1cm} {\tt Vertical Trip}

\hspace*{1cm} take port $W$;

\hspace*{1cm}  {\bf while} no pebble on current node {\bf do}

\hspace*{2cm} take port $W$;

\hspace*{1cm} {\tt Move Pebbles}

\vspace*{0.5cm}

Now the algorithm can be succinctly formulated as follows.

\vspace*{0.5cm}

{\bf Algorithm} {\tt Explore Large Semi-Walled} 

\hspace*{1cm} {\tt Modified Small Free Exploration}

\hspace*{1cm} {\bf if} $below=false$ {\bf then}

\hspace*{2cm} {\tt Slide Down to Origin - Above}

\hspace*{2cm} {\tt Initializing Pebbles - from Above}

\hspace*{1cm} {\bf else}

\hspace*{2cm} {\tt Slide Down to Origin - Below}

\hspace*{2cm} {\tt Initializing Pebbles - from Below}

\hspace*{1cm} {\bf repeat forever}

\hspace*{2cm} {\tt Phase}

\vspace*{0.5cm}

\vspace*{0.5cm}

Since procedure {\tt Modified Small Free Exploration} guarantees hitting the wall and the rest of the algorithm guarantees the exploration of the entire grid in the presence of the wall, we conclude that the algorithm correctly explores all large semi-walled wedges. All parameters used in the procedures can be precomputed knowing the slope of the wall. Hence the algorithm can be executed by an automaton. Since two pebbles are used, we have the main positive result of this section.

\begin{theorem}\label{large semi-walled}
For any large semi-walled wedge there exists an automaton that explores it with two pebbles.
\end{theorem}

\subsubsection{Impossibility of exploration with one pebble}

In this section we prove that the number of pebbles from Theorem \ref{large semi-walled}  cannot be decreased.

\begin{theorem}
No large semi-walled wedge can be explored by an automaton with one pebble.
\end{theorem}
 
\begin{proof}
Consider any automaton equipped with one pebble and let $T_v$ be its trajectory in the empty grid (without any walls) starting at node $v$.
By the argument from the proof of Theorem \ref{lb small free} there exist two parallel lines $L'(v)$ and $L''(v)$ such that the entire trajectory $T_v$ is contained between those two lines and the distance between them does not depend on $v$.
Now consider a large semi-walled wedge $W$ with wall $H$. Since $W$ is large,  there exists a node $v$ in $W$ such that  the part of the plane between lines $L'(v)$ and $L''(v)$ is at distance larger than 1 from any point of the half-line $H$. Consequently the agent starting at node $v$ will never hit the wall, and thus cannot explore the entire wedge $W$.
\end{proof}

\begin{corollary}
The minimum number of pebbles sufficient to explore any large semi-walled wedge is two.
\end{corollary}

\subsection{Free wedges}

In this section we show that the minimum number of pebbles sufficient to explore a large free wedge is three.

\subsubsection{Exploration with three pebbles}

We first show how to explore any large free wedge using three pebbles. It is enough to consider the hardest case, when the wedge is the entire grid. The high-level idea of the algorithm is the following.
Pebbles are initialized so that  the first pebble is dropped at the starting node $v$ of the agent, the second pebble is one step North of $v$ and the third pebble is one step East of $v$. The agent is back at $v$. 
Then this elementary right triangle is expanded in consecutive phases, pushing and exploring each of its sides by one step. Suppose that the old positions of the three pebbles were $P_1,P_2,P_3$, with $P_1$ at the right angle and $P_2, P_3$ at distance $x$ steps North and East of it, respectively, and all grid nodes of this triangle were explored. Then the new position $P'_1$ of the first pebble is at distance 2 from  $P_1$, South-West from it, the new position $P'_2$ of the second pebble is at distance 3 from  $P_2$, North-North-West from it, and  the new position $P'_3$ of the third pebble is at distance 3 from  $P_3$, South-East-East from it. Vertical and horizontal sides of the triangle are explored going straight, and the hypothenuse side is explored using a staircase (see Fig. \ref{fig-6}). Hence all grid nodes of the triangle $P'_1,P'_2,P'_3$ are explored.

\begin{figure}[h]
   \vspace*{-.194in}
     \centering
    \includegraphics[scale = .6]{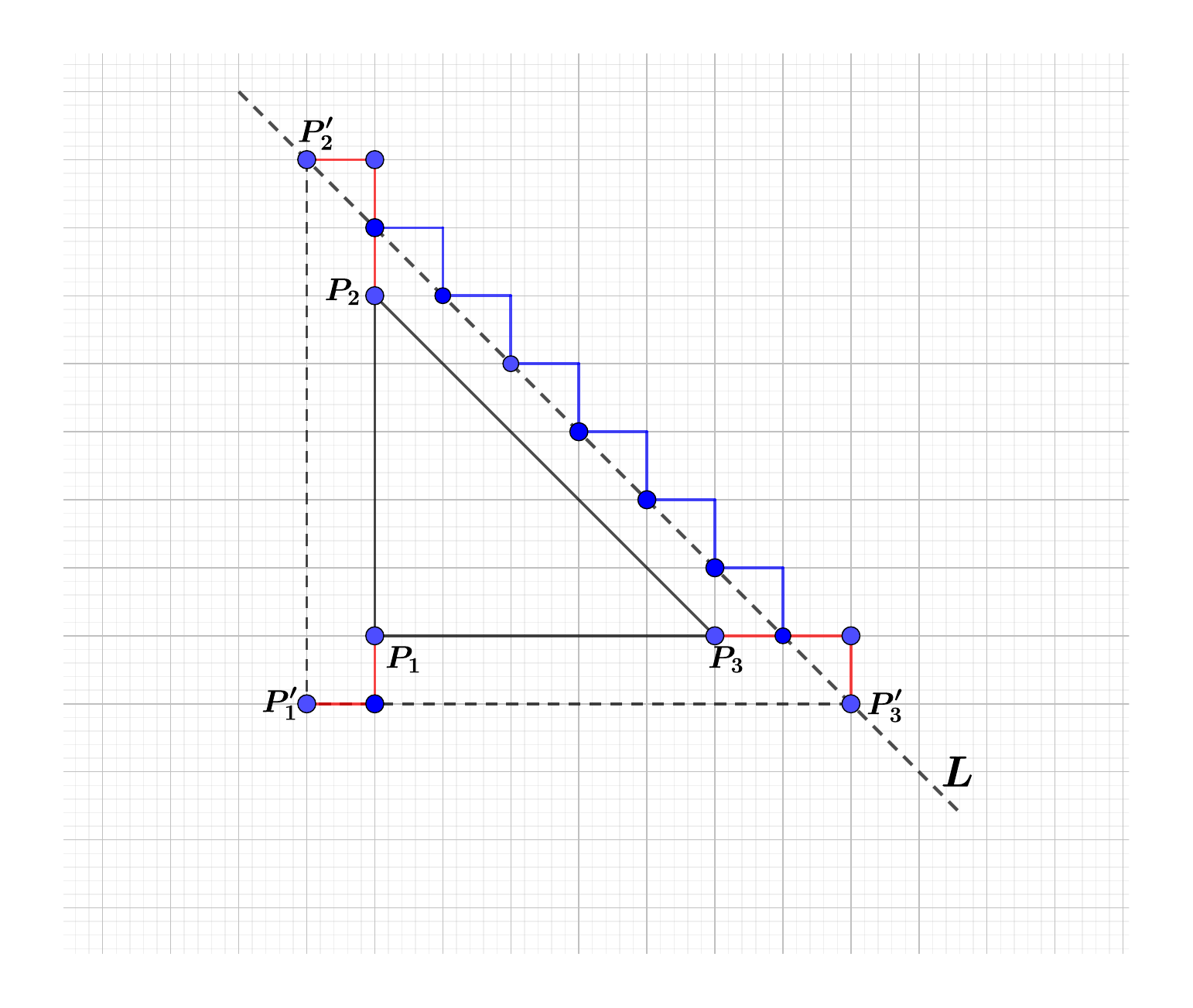}
   \vspace*{-.2in}
    \caption{A phase in the procedure {\tt Expand Triangle}. At the beginning of the phase the pebbles are at nodes $P_1,P_2,P_3$ and the agent is at $P_1$. 
    At the end of the phase the pebbles are at nodes $P'_1,P'_2,P'_3$, the agent is at node $P'_1$ and it has explored the triangle $P'_1,P'_2,P'_3$. }
    \label{fig-6}
\end{figure}

The following algorithm implements the above idea.

\pagebreak

{\bf Algorithm} {\tt Expand Triangle}

\hspace*{1cm} drop a pebble;

\hspace*{1cm} take port $N$;

\hspace*{1cm} drop a pebble;

\hspace*{1cm} take port $S$;

\hspace*{1cm} take port $E$;

\hspace*{1cm} drop a pebble;

\hspace*{1cm} take port $W$;

\hspace*{1cm} {\bf repeat forever}

\hspace*{2cm} pick a pebble;

\hspace*{2cm} take port $S$;

\hspace*{2cm}  take port $W$;

\hspace*{2cm} drop a pebble;

\hspace*{2cm} take port $E$;

\hspace*{2cm} {\bf while} no pebble on current node {\bf do}

\hspace*{3cm} take port $N$;

\hspace*{2cm} pick a pebble;

\hspace*{2cm} take port $N$;

\hspace*{2cm} take port $N$;

\hspace*{2cm}  take port $W$;

\hspace*{2cm} drop a pebble;

\hspace*{2cm}  take port $S$;

\hspace*{2cm} {\bf while} no pebble on current node {\bf do}

\hspace*{3cm}  take port $S$;

\hspace*{2cm} take port $E$;

\hspace*{2cm} take port $N$;

\hspace*{2cm} {\bf while} no pebble on current node {\bf do}

\hspace*{3cm} take port $E$;

\hspace*{2cm} pick a pebble;

\hspace*{2cm} take port $E$;

\hspace*{2cm} take port $E$;

\hspace*{2cm} take port $S$;

\hspace*{2cm} drop a pebble;

\hspace*{2cm}  take port $W$;

\hspace*{2cm} {\bf while} no pebble on current node {\bf do}

\hspace*{3cm} take port $W$;

\hspace*{2cm} take port $E$;

\hspace*{2cm} {\bf while} no pebble on current node {\bf do}

\hspace*{3cm} take port $E$;

\hspace*{2cm} take port $N$;

\hspace*{2cm} take port $W$;

\hspace*{2cm} {\bf while} no pebble on current node {\bf do}

\hspace*{3cm} take port $N$;

\hspace*{3cm} take port $W$;

\hspace*{2cm} take port $S$;

\hspace*{2cm}  {\bf while} no pebble on current node {\bf do}

\hspace*{3cm} take port $S$;

\vspace*{0.5cm}

The description of Algorithm  {\tt Expand Triangle} implies the following result.

\begin{theorem}\label{th-large-free}
For any large free wedge there exists an automaton that explores it with three pebbles.
\end{theorem}

\subsubsection{Impossibility of exploration with two pebbles}

It remains to show that the number of pebbles in Theorem \ref{th-large-free} cannot be decreased. This follows 
from the result of \cite{BUW} which says that 3 (semi-synchronous) deterministic automata cannot explore the whole grid. In fact, the proof from \cite{BUW} shows more:  3 (semi-synchronous) deterministic automata cannot explore even all grid nodes in a half-plane, i.e., (in our terminology) cannot explore any large free wedge. Since automata can simulate pebbles, this proves that an automaton with two pebbles cannot explore any large free wedge.

\begin{corollary}
The minimum number of pebbles sufficient to explore any large free wedge is three.
\end{corollary}

\section{Conclusion}

We completely solved the problem of determining the minimum number of pebbles sufficient to explore any wedge of the oriented grid by an automaton, depending on the type and size of the wedge. These results yield several open problems. What is the minimum number of pebbles sufficient to explore other natural infinite subgraphs of the oriented grid by an automaton, for example the grid with finitely many obstacles? A natural generalization is to infinite subgraphs of multi-dimensional grids. In \cite{BS77} it was proved that no finite collection of automata can explore all connected subgraphs of the 3-dimensional grid, but the question is still open if an infinite subgraph of this grid is given as input, as we do in this paper.  The same question can be asked about infinite graphs other than subgraphs of a grid.

%%%%%%%%%%%%%%%%%%%%%%%%%%%%%%%%%%%%%%%%%%%%%%%%%%%%%%%%%%%

%%%%%%%%%%%%%%%%%%%%%%%%%%%%%%%%%%%%%%%%%%%%%%%%%%%%%%%%%%%

%%%%%%%%%%%%%%%%%%%%%%%%%%%%%%%%%%%%%%%%%%%%%%%%%%%%%%%%%%%
\end{document}